\documentclass[aps,nofootinbib]{revtex4-1}
\usepackage[english]{babel}
\usepackage[applemac]{inputenc}
\usepackage{bbm,epigraph}
\usepackage{amsfonts}
\usepackage{amsthm,comment}
\usepackage{color}
\usepackage{graphicx,MnSymbol}
\usepackage[all,cmtip]{xy}

\newcommand{\inp}{	\,\ensuremath{\raise-.5ex\hbox{$\invbackneg$}}} 
\newcommand{\qsp}[2]{\,\ensuremath{\raise.5ex\hbox{$#1$}\big\slash\raise-.5ex\hbox{$#2$}}} 

\newcommand{\Tr}[1]{\mathrm{Tr}\left[#1\right]}

\newcommand{\ev}[1]{\emph{#1}}

\newcommand{\dl}{\mathsf{d}_\ell}

\newcommand{\R}{\mathbb{R}}
\newcommand{\Ad}{\mathsf{Ad}}
\newcommand{\Cad}{\mathsf{Ad}^*}
\newcommand{\cad}{\mathsf{ad}^*}
\newcommand{\ad}{\mathsf{ad}}
\newcommand{\orb}{\mathcal{O}}

\newtheorem{definizione}{Definition}[section]
\newtheorem{theorem}[definizione]{Theorem}

\newtheorem{proposition}[definizione]{Proposition}
\newtheorem{lemma}[definizione]{Lemma}

\theoremstyle{definition}

\bibstyle{apsrev4-1}

\newenvironment{definition}{\begin{definizione}}{\end{definizione}}

\newtheoremstyle{Rem}
{}
{}
{}
{}
{\itshape}
{:}
{\newline}
{}
\theoremstyle{Rem}
\newtheorem*{Remark}{\textsf{Remark}}

\newcommand{\T}[2]{\mathsf{T}_{#2}#1}

\newenvironment{remark}{\begin{Remark}}{ \end{Remark}}

\begin{document}

\title[Geometric Fisher Information Tensor]{On the geometry of mixed states and the Fisher information tensor}

\date{\today}

\author{I. Contreras}
\affiliation{Department of Mathematics, University of Illinois at Urbana-Champaign, 1409 W. Green Street, IL 61801 Urbana, USA}
\email{icontrer@illinois.edu}


\author{E. Ercolessi}
\affiliation{Dipartimento di Fisica e Astronomia, Universit\`a di Bologna and INFN, V. Irnerio 46, 40127 Bologna, Italy}
\email{ercolessi@bo.infn.it}

\author{M. Schiavina}
\affiliation{Institut f\"ur Mathematik, Winterthurerstrasse 190, 8057 Z\"urich, Switzerland}
\email{michele.schiavina@math.uzh.ch}

\thanks{This research was (partly) supported by the NCCR SwissMAP, funded by the Swiss National Science Foundation. I.C. acknowledges support from the SNF Grant $P300P2\_154552$, and from the Department of Mathematics of the University of California at Berkeley, where most of this research was done. E.E. thanks G. Marmo for interesting discussions and acknowledges financial support from the INFN grant QUANTUM. M.S. acknowledges partial support from the SNF grant $200020\_149150/1$. I.C. and M.S. are thankful to the organisers of Poisson '14 for the support they were granted, as some of the results presented here have been worked out during the conference.}

\begin{abstract}
In this paper we will review the co-adjoint orbit formulation of finite dimensional quantum mechanics, and in this framework we will interpret the notion of  quantum Fisher information index (and metric). Following previous work of part of the authors, who introduced the definition of Fisher information tensor, we will show how its antisymmetric part is the pullback of the natural Kostant Kirillov Souriau symplectic form along some natural diffeomorphism. In order to do this we will need to understand the symmetric logarithmic derivative (SLD) as a proper 1-form, settling the issues about its very definition and explicit computation. Moreover, the fibration of co-adjoint orbits, seen as spaces of mixed states, is also discussed.
\end{abstract}


\maketitle

\section*{Introduction}
The idea of understanding finite dimensional quantum mechanics in a geometric framework \cite{morandi} is a fairly established approach to quantum problems that has lead to many interesting results and new perspectives both on the foundations of quantum mechanics as well as on applications. Quantum non-local phenomena like the Ahronov Bohm effect \cite{AB} or the existence of topological phases of matter, have been  predicted first and then observed. The non-local features of the phenomena under study are in fact well described in a geometric framework, where the topology of the configuration space (holes, defects, boundaries, etc.) plays a role in determining the quantum behaviour of the physical system. More recently, a series of studies \cite{Beng, Christandl} have suggested that we look at the entanglement of quantum states, probably the most famous non-local feature of QM, from a geometric perspective. The idea is to consider entanglement classes under the action of $SL_n(\mathbb{C})$ and intersect them in some suitable sense with the Weyl chamber in the Lie algebra $\mathfrak{sl}_n(\mathbb{C})$. This point of view lead to the definition of what is called the \emph{entanglement polytope}, and served in understanding how the foundational Pauli exclusion principle can be extended. The approach therein is strikingly similar to the one presented here, and this vicinity calls for a deeper understanding.

In the context of quantum metrology, a great deal of research has been dedicated to the problem of computing distances between quantum states, optimizing quantum measurements, and to the relationship that there is between these problems and quantum information. In fact, there exist several metrics in the \emph{space of states} \cite{Petz,Zyk}, among which we find the quantum Fisher information metric, that turns out to play an important role in information theory and quantum estimation theory \cite{metro}. It is a well established fact at this point that the quantum Fisher metric for pure states coincides (up to a scalar) with the Fubini-Study metric on the associated projective spaces \cite{Brody, Marmo} (and references therein). 

The computation of such a metric and the precise understanding of it in a wider picture is essential to unlock further developments, both on the experimental and on the theoretical side. As a matter of fact, the Fisher information metric relies on the computation of the symmetric logarithmic derivative (SLD), as suggested by \cite{Braun,Burb}, and until recently its explicit computation was possible only for pure states \cite{Barn} and two dimensional mixed states \cite{luati1,luati2}.

Various contributions, first of \cite{Marmo} and followed by \cite{ES1, ES2}, have shown that a geometric approach can be fundamental in understanding the objects that come into play, helping to address their computation and interpretation. Specifically, the key idea (see also \cite{ercolessi}) is to look at the space of states (pure or mixed) as co-adjoint orbits of the group of unitary matrices. Such a simple point of view enables us to use powerful tools from Lie theory and differential geometry. This has lead to the reduction of the problem of computing the SLD for an arbitrary $n$-dimensional mixed states to a system of $n^2$ linear equations, avoiding the decomposition of the matrix of derivatives $\partial \rho$ in a basis of eigenstates, allowing in this way the explicit computation of the Fisher metric for a mixed three level system in local coordinates \cite{ES2}, and potentially for any mixed state. Moreover, in \cite{ES1} the Fisher metric has been shown to coincide with the round metric on $S^3$ when the variation of the parameters is taken into account, i.e. transversally to a given co-adjoint orbit.

With the present paper we wish to develop this framework in greater detail, and make it accessible to a broader audience, including mathematicians and physicists interested in quantum metrology, as well as in differential- geometric and Lie theoretical aspects of quantum systems. Although some of the key mathematical techniques used in the sequel are widely known to experts in the field, they are not very common in the physics literature. We will present the basics of the theory of co-adjoint orbits and of compact Lie theory necessary to our purposes, to usher the reader less familiar with the formalism, into a deeper understanding of the mathematical framework. In section \ref{QM} we will make contact between the general mathematical theory and finite dimensional quantum mechanics. 

The added value comes from the translation of the physical data into a precise geometric language. This will reward us with new results and many insights on new interesting areas to be explored. The core of this paper will be devoted to showing how the symmetric logarithmic derivative can be defined rigorously, by constructing a natural vector bundle morphism that promotes solutions of the equation
\[
A=\frac12\{X,D\}
\] 
for $D$ a suitable diagonal matrix and $A$ an Hermitian off-diagonal matrix, to Lie algebra valued 1-forms. To emphasise this new point of view we will rename the SLD as \emph{symmetric logarithmic differential}. 

The problem with the above equation, and consequently with definition of SLD in the physics literature, is that it is not clear how to deal with solutions to the homogeneous equation $\{X,D\}=0$, and the uniqueness of the SLD. Moreover, the particular solution itself, when the problem is approached in the Hilbert space formalism, becomes utterly cumbersome, and one has to guess a correct \emph{ansatz} for the solution \cite{luati1}, even when the variation of the matrix $\rho$ is allowed only along a curve $\rho(\theta)$. 

This new perspective will allow us to interpret the re-defined \emph{symmetric logarithmic differential} $\dl\rho$ as  the \emph{unique} $U(n)$-equivariant Lie algebra valued 1-form that is obtained from the tautological 1-form $d\rho$ by pullback along an appropriate vector bundle map. This, in turn, will allow us to prove that the antisymmetric part of the so called  Fisher tensor, as defined in \cite{ES1,ES2}, is the pullback of the natural Kirillov-Kostant-Souriau (KKS) symplectic form on co-adjoint orbits, and that furthermore it is symplectic (Theorem \ref{propantisymm}, page \pageref{propantisymm}). The symmetric part of the Fisher tensor yields the quantum Fisher information metric as expected, and the question of whether this procedure preserves the K\"ahler structures on co-adjoint orbits arises naturally.  

As we mentioned above, the geometric features studied in this paper open up several gates towards a better understanding of the relationship between finite dimensional quantum systems and the geometry of co-adjoint orbits. One of them is the comparison of the symplectic structure for different co-adjoint orbits. In particular, one can show that the fibration of a flag manifold over smaller co-adjoint orbits is a symplectic fibration (in the sense of \cite{Guill}), with the generic fibre being a flag manifold of lower dimension (or products thereof). This implies the existence of a symplectic connection compatible with the KKS symplectic form for different orbits. We will illustrate how to interpret fibrations of co-adjoint orbits, in terms of a nesting of the spaces of mixed states. This will clarify similarities and differences between the pure and the mixed state cases, and it will help addressing the problem of finding a unified framework to treat them consistently.  

A natural question that arises here is whether there exists a similar symplectic connection compatible with the symplectic structure that appears from the Fisher tensor. In a subsequent work we intend to elaborate on this issue, providing a way to understand the symmetric logarithmic differential in terms of the symplectic fibration. Moreover, this hierarchy of co-adjoint orbits is very likely to be related to the hierarchy of entanglement orbits in \cite{Christandl}, and more research should be done to understand this relation.

Furthermore, it is interesting to understand the relationship of these constructions with K\"ahler geometry. It is known \cite{GuSt,Kost}, that there are ways to equip the co-adjoint orbits with a complex structure compatible with the symplectic structure, and choose polarisations that are essential in the context of geometric quantisation. The \emph{twist} introduced by the logarithmic differential, and subsequently by the Fisher tensor, might be reflected in the choice of polarisation and K\"ahler structure, and further research can surely be considered in this direction.

In Section \ref{Mixed} we will cast finite dimensional quantum mechanics (QM), in the density matrix representation, as the theory of co-adjoint orbits for the unitary group $U(n)$. We will review the concepts of pure and mixed states in this new setting, making contact with the usual Hilbert space description.

In Section \ref{Fibration} we will describe the fibration of co-adjoint orbits in the case of QM, as a nesting of spaces of states, introducing the main objects and setting the basis for the subsequent constructions. 

Section \ref{SYLODE} is devoted to make the notion of symmetric logarithmic differential geometrically precise, and the main constructions that will be needed for the main results will be expounded. 

Finally, in Section \ref{FishT} we will use the outlined constructions to understand the geometric properties of the Fisher Tensor, as defined in \cite{ES2}. We will prove how its antisymmetric part can be seen as the pullback of the Kirillov-Kostant-Souriau symplectic form, while its symmetric part is closely related to the quantum Fisher information index and with the K\"ahler structures induced by the Hermitian product on $\mathfrak{u}^*(n)$.

Appendices \ref{App:coad} and \ref{App:generalities} will contain a review of some basic literature results on co-adjoint orbits and compact Lie groups, that will be used throughout the paper.


\section{Spaces of mixed states} \label{Mixed}
In this section we will review some aspects of the theory of co-adjoint orbits for the case $G=U(n)$, the group of unitary $n$-dimensional matrices. This will allow us to set finite dimensional quantum mechanics in a geometrically precise framework, starting from the usual Hilbert space description and generalizing the projective Hilbert space description of density matrices. A review of the general theory of co-adjoint orbits for general Lie groups can be found in Appendix \ref{App:coad}.

\subsection{The unitary group and quantum mechanics}\label{QM}
The classical construction of symplectic manifolds from the orbits of the co-adjoint action works for any Lie group (see Appendix \ref{App:coad}), yet it is worthwhile to have a closer look at the case $G=U(n)$, the compact group of $n\times n$ unitary matrices. As we will see,  such orbits  turn out to be the spaces of states for ordinary finite dimensional quantum mechanics, in the density matrix representation.  


The Lie algebra $\mathfrak{u}(n)$ of $U(n)$ given by $n\times n$ anti-Hermitian matrices has a maximal abelian subalgebra $\mathfrak{t}\subset\mathfrak{u}(n)$ given by  diagonal pure imaginary matrices. Also, the dual space $\mathfrak{u}^*(n)$ can be identified with the Lie algebra $i\mathfrak{u}(n)$ of $n\times n$ Hermitian matrices 
thanks to the compactness of $\mathfrak{u}(n)$, which also allows to identify the adjoint and co-adjoint representations. More specifically, here and in the following we make use of the pairing
\begin{equation}
\langle g , m\rangle := i Tr[gm] \;\; , g\in \mathfrak{u}^*(n), \, m\in \mathfrak{u}(n)
\end{equation}
which is equivariant with respect to the (co)-adjoint action.

Given a choice of ordering of the eigenvalues of matrices in $\mathfrak{t}$ (see Appendix \ref{App:generalities}), one has a one-to-one correspondence between elements in the chosen positive Weyl chamber $i\mathcal{C}^+\subset\mathfrak{t}$ and the set of all (co-)adjoint orbits \cite[Lemma 3, Chapter 5, Section 2.1]{Kir}. So the elements $\rho_0\in\mathfrak{t}$ are the preferred reference points for the orbits, and the classification of these orbits depends solely on the multiplicity of the eigenvalues of $\rho_0$. For instance, the choice $\rho_0=\mathrm{diag}_n\{1,0,\dots,0\}$ will give rise to the orbit 
\[
\orb_{\rho_0}\simeq\qsp{U(n)}{U(1)\times U(n-1)}.
\]

The precise geometric construction and properties will be expounded in the following sections, but we would like to give here an intuition and a physical motivation for this construction.

The usual setting for  finite dimensional quantum mechanics is a Hilbert space $H$,  with \emph{states} being represented by vectors $\psi\in H$.  Such vectors usually arise as eigenvectors of a given reference observable $\mathsf{O}$ (a self adjoint operator), describing for instance the possible outcome of a measurement. All observable quantities in quantum mechanics are nevertheless dependent on less information than the vector $\psi$ itself: they depend neither on  its $\mathbbm{R}^+$-valued norm nor on its overall $U(1)$-phase. For this reason, the correct space of states to consider is rather the projective Hilbert space 
\begin{equation}\label{projhilb}
P(H)\simeq\qsp{H^\times}{\mathbbm{R}^+\times U(1)}\simeq\qsp{S^{2n-1}}{U(1)}\simeq\qsp{U(n)}{U(n-1)\times U(1)}
\end{equation} 
where $H^\times = H -\{0\}$.

This construction generalizes to situations in which the system is not anymore described by a single vector in a Hilbert space, but rather by an orthonormal basis $\{\psi_i\}_{i=1\dots k}$ for some $k$-dimensional subspace of $H$. In this case we speak of \emph{mixed states}, and the construction of the space of states becomes more involved. The \emph{generalized spheres}, taking the place of the odd spheres $S^{2n-1}$ appearing in Eq. \eqref{projhilb}, will be given by the Stiefel manifolds for the complex field 
\begin{equation}
\mathcal{B}^{n,k}:=\qsp{U(n)}{U(n-k)} \label{bkn}
\end{equation}
since $\mathcal{B}^{n,k}$ is the $(2nk-k^2)$-dimensional manifold of ordered $k$-tuples of orthonormal vectors in $\mathbb{C}^n\simeq H$ or, said in equivalent way, the space of all matrices $g \in \mathfrak{u}^*(n)$ with fixed rank $k$. 

Assume that the multiplicities of the $\psi_i$'s as eigenvectors of $\mathsf{O}$ are all equal to $1$, meaning that the eigenvalues are distinct. Then, to account for a phase degeneracy for each $\psi_i$ we must quotient out $\mathcal{B}^{n,k}$ by a $U(1)$-factor for each vector in the basis, obtaining:
\begin{equation}
\mathcal{P}^{n,k}:=\qsp{U(n)}{U(n-k)\times U(1)^k}.
\end{equation}
Here $\mathcal{P}^{n,k}$ represents the set of rank $k$ Hermitian matrices with the same distinct, non-zero eigenvalues. 
In  this quotient we are identifying matrices that  are isospectral and share the same eigenspaces, which are one-dimensional in this example. 


If we admit now a degenerate eigenvalue with  multiplicity $d$, the counting of possible choices for the orthonormal system  must take into account to the $U(d)$-action that reshuffles the eigenvectors in the corresponding eigenspace. 
Thus, in place of $\mathcal{P}^{n,k}$, we get the coset
\begin{equation}\label{generalcoset}
\qsp{U(n)}{U(n-k)\times U(d)\times U(1)^{k-d}}.
\end{equation}
Notice that the integers $m_i$ appearing in the $U(m_i)$ factors in the quotient form a partition of $n$, that is to say: $m=\{m_i\}=\{n-k, d, 1,\dots,1\},\ m_i>0,\ \sum_im_i=n$.

The viewpoint we have presented here will be replaced by an equivalent description in terms of  \emph{density matrices} (cf. Definition \ref{Def.mixedstate}). We will see that coset spaces such as \eqref{generalcoset} arise naturally and in a straightforward way.

\subsection{Mixed states}
In section \ref{QM} we discussed how a description of quantum mechanics in the Hilbert space formalism becomes involved when we want to include mixed states in the picture. It is possible to  adopt another point of view \cite{ercolessi} and  greatly simplify some aspects of the geometric description of quantum states.  In order to do this we have to clarify what we mean by mixed state and introduce some definitions.

\begin{definition}\label{Def.mixedstate}
A \ev{mixed state} in an n-dimensional Hilbert space is specified by a \ev{density matrix}, i.e. by a Hermitian matrix $\rho$ such that $\Tr{\rho}=1$ and $\rho$ is positive definite.  Regarding $\rho$ as an element of $\mathfrak{u}^*(n)$, we will denote by ${\mathcal{D}}\subset \mathfrak{u}^*(n)$ the set of all mixed states. 
\end{definition}

\begin{remark}
The set ${\mathcal{D}}^k\subset {\mathcal{D}}$ of all rank $k$ mixed states, can be identified with the intersection of the positive cone in the space $B^n_k$ defined in (\ref{bkn}) with the affine subspace of unit trace matrices.  The $\mathcal{D}^k$'s are also called \emph{strata}. 

It has been shown \cite{grab} with techniques that are different from the ones presented here, that the stratum ${\mathcal{D}}^k$ is a smooth connected submanifold of $\mathfrak{u}^*(n)$ for any $k$ (of dimension $2nk-k^2-1$). Furthermore, the stratification of the different $\mathcal{D}^k$ in $ \mathcal{D}$ is  \emph{maximal} (for $k>2$), in the sense that any smooth curve in  $\mathfrak{u}^*(n)$ that  belongs entirely to $ {\mathcal{D}}$  actually lies on a single stratum ${\mathcal{D}}^k$. This implies that, also from a differentiable point of view, we can restrict  our attention to  a single ${\mathcal{D}}^k$, i.e. to the space of density matrices with fixed rank. 

We observe furthermore that $\rho$ is actually a \emph{pure state} if and only if it is a rank one projector, i.e. if it satisfies the additional property $\rho^2=\rho$.
\end{remark}

Now, let $\rho\in  {\mathcal{D}}$. From the spectral theorem we know that there exists a unitary matrix $U\in U(n)$ and an \ev{$\mathbf{n}$-decomposition of unity} $\kappa=\{k_1,\dots,k_n\}$, $\sum k_i=1$ and $k_i\geq 0\ \forall i$, such that
\begin{equation}
\rho=\Ad_U \rho_0,\ \ \ \rho_0=\mathrm{diag}(\kappa)\equiv\mathrm{diag}\{k_1,\dots,k_n\}.
\end{equation} 
Therefore, to each $\rho\in\mathcal{D}$ we associate a unique partition of $n$ by counting the multiplicities $m_\kappa=\{m_i\}$ of the eigenvalues in $\kappa$. Let us consider the  (co-)adjoint orbit   passing through a given density matrix $\rho$ in the Lie algebra $\mathfrak{u}^*(n)$. As we said,  such a mixed state is determined by $\kappa$, a decomposition of unity, and by $m_\kappa$, its relative partition of $n$. Its stabiliser subgroup $H_\rho$ depends solely on the partition $m_\kappa$. As a matter of fact
\begin{equation}
H_\rho\simeq\bigtimes\limits_{m_i\in m_\kappa }U(m_i).
\end{equation}

Then, we may associate with a mixed state $\rho$ the adjoint orbit given by
\begin{equation}
\mathcal{O}_{\rho_0}\mathrm{:=}\{\rho\in i\mathfrak{u}(n)\ \big|\ \rho=\Ad_U\rho_0,\ U\in U(n)\}\simeq\qsp{U(n)}{\bigtimes\limits_{m_i\in m_\kappa}U(m_i)}.
\end{equation}

The \emph{topological type} of quotient  depends only on the stabiliser subgroup, which in turn depends on the partition of $n$ induced by the eigenvalues of $\rho_0$. The actual values of $\kappa$ select a particular orbit, which then represents the space of states  with \emph{fixed eigenvalues} $\kappa$ and corresponding eigenspaces. The space of orbits of the same topological type will represent the space of states with $m_\kappa$ eigenvalue degeneracy, thus preserving only the eigenspace structure. 

\begin{remark}
To clarify the difference between the topological type of an orbit and a choice of one particular instance, consider the case of $U(2)$. Each coadjoint orbit is a sphere, as it is isomorphic to the coset $\qsp{U(2)}{U(1)\times U(1)}$. The particular choice of an element in the (positive) Weyl chamber selects a given radius for the sphere, while the union of all co-adjoint orbits such that $\Tr{\rho}=1$ yields the unit ball in $\mathfrak{u}(2)$. In the mixed states language this is sometimes called the Bloch ball \cite{Preskill}.

Notice that this distinction between orbits wouldn't be so relevant if we discarded the information on the complex structure and the compatibility with the symplectic structure. Metrically different orbits are nevertheless to be distinguished for our purposes.
\end{remark}

\begin{definition}
We define the \ev{space of $\kappa$-mixed states} to be the co-adjoint orbit of the unitary group passing through (or represented by) the diagonal element $\rho_0=\mathrm{diag}\{\kappa\}$.
\end{definition}

The space of states will be representing all possible states \emph{with the same mixing}. This is equivalent to asking how many different orthonormal frames one can find requiring that the basis vectors be eigenvectors of a  given operator, with { multiplicity structure given by a decomposition of unity $\{m_\kappa\}$}. 

Such a space of $\kappa$-mixed states is a straightforward generalization of the projective space of pure states.

\begin{remark}
Notice that there is no preferred choice of ordering for the basis of eigenvectors, and we could assume that the eigenvalues with the least multiplicity are put first. For instance, if we are given the ``unordered'' list $\widetilde{\kappa}=\{\mu, \zeta, \mu, \mu, \alpha, \beta,\beta, \nu\}$ we could consider instead the ordered list $\kappa=\{\zeta,\alpha, \beta, \beta, \mu,\mu,\mu \}$ and the respective partition of unity would then be $m_\kappa=<1,1,2,3>$. Another very common choice it to require that the eigenvalues be weakly increasing.

More precisely, this ordering is associated with an ordering in the root system, and to the choice of a preferred Weyl chamber in the subalgebra $\mathfrak{t}$ (see Appendix \ref{App:generalities}).
\end{remark}

\section{Fibration} \label{Fibration}
In this section we will analyse the rich geometric structure that co-adjoint orbits have in relation with one another. In general, under some assumptions that are very mild when looking at compact groups, generic orbits fibrate over lower dimensional ones. In the case of quantum mechanics we may interpret this fibration as a nesting of spaces of states, and this nesting is compatible with the geometric structure.

In this section the Lie group $G$ will be always required to be compact and, further on, we will require $G=U(n)$ in order to make contact with the theory of mixed states.

\subsection{General theory}\label{genth}
As explained in Appendix \ref{App:coad}, Eq. \eqref{Pbdl}, each co-adjoint orbit can be seen as the base space of a principal bundle, the structure group {  varying} according to the base point through which the orbit is considered. We can choose the representative of the orbit to be an element in the \emph{Cartan subalgebra} $\mathfrak{t}\subset\mathfrak{g}$. Different points in the Cartan subalgebra $\xi_\bullet,\eta_\bullet\in\mathfrak{t}$, will describe different orbits, whose topological type varies depending on the the stabiliser of $\xi_\bullet$ (resp $\eta_\bullet$), and which are diffeomorphic if and only if the stabiliser subgroups ($H_{\xi_\bullet}$, resp. $H_{\eta_\bullet}$) are diffeomorphic.

Let $\xi_\bullet$ and $\eta_\bullet$ be such that their stabilizing sub-algebras are ordered: $\mathfrak{h}_{\eta_\bullet}\subset\mathfrak{h}_{\xi_\bullet}$, i.e. $H_{\eta_\bullet}$is a closed subgroup of $H_{\xi_\bullet}$. Then, it is possible to consider the fibre bundle associated to the principal bundle
\begin{equation}
H_{\xi_\bullet}\rightarrow G \rightarrow \mathcal{O}_{\xi_\bullet}
\end{equation}
given by \cite{Guill}
\begin{equation}
G\times_{H_{\xi_\bullet}}\left(\Ad_{H_{\xi_\bullet}} \eta_\bullet\right)\simeq \qsp{\left(G\times\left(\Ad_{H_{\xi_\bullet}} \eta_\bullet\right)\right)}{H_{\xi_\bullet}}\simeq\mathcal{O}_{\eta_\bullet}
\end{equation}
which is diffeomorphic to the orbit through the point $\eta_\bullet$. This means that the \emph{bigger} orbit $\mathcal{O}_{\eta_\bullet}$ is the total space of a fibre bundle over $\mathcal{O}_{\xi_\bullet}$, with fibre  diffeomorphic to a co-adjoint orbit of {$H_{\xi_\bullet}$, namely:}
\begin{equation}\label{Seq}
\qsp{H_{\xi_\bullet}}{H_{\eta_\bullet}}\rightarrow \mathcal{O}_{\eta_\bullet}\rightarrow \mathcal{O}_{\xi_\bullet}
\end{equation}

The projection that defines the fibration is given by:
\begin{equation}
\pi(\eta)=\pi(\Ad_g\eta_\bullet) = \Ad_g\xi_\bullet
\end{equation}
for $\eta\in\mathcal{O}_{\eta_\bullet}$ and $g\in G$ , while the tangent map $\pi_*$ is given by
\begin{equation}
\pi_*\colon\big|_{\eta_\bullet}\begin{array}{rcl}
T_{\eta_\bullet}{\mathcal{O}_{\eta_\bullet}} & \longrightarrow & T_{\xi_\bullet}{\mathcal{O}_{\xi_\bullet}}\\
X_{\eta_\bullet}&\longmapsto& \ad_{X_{\eta_\bullet}}\xi_\bullet.
\end{array}
\end{equation}

Recall that compact Lie groups allow for a canonical identification of the Lie algebra and its dual, and  by means  of the $Ad$-invariance of the Killing form, and whenever $\xi_\bullet$ is a split point, i.e. the Lie algebra splits (equivariantly) as
\begin{equation}\label{split}
\mathfrak{g}=\mathfrak{h}_{\xi_\bullet} \oplus \mathfrak{n}_{\xi_\bullet}
\end{equation}
we can choose an affine subspace of $\mathfrak g$ of the form
\begin{equation}
N_{\xi_\bullet}:=\xi_\bullet + \mathfrak {h}_{\xi_\bullet}^{\perp}
\end{equation} 
requiring that it be transversal to the orbit $\mathcal O_{\eta_{\bullet}}$ for some $\eta_\bullet\in N_{\xi_\bullet}$. Following \cite[Theorem 2.3.3]{Guill}, this choice is equivalent to finding a (symplectic) connection for the fiber bundle \eqref{Seq}. In particular, for $G$ compact, every point $\eta_\bullet$ in $N_{\xi_\bullet}$ is a split point and the transversality condition is always satisfied, so we really only have to check that $\mathfrak{h}_{\eta_\bullet}\subset\mathfrak{h}_{\xi_\bullet}$ or equivalently that $\mathfrak{n}_{\eta_\bullet}\supset\mathfrak{n}_{\xi_\bullet}$, where both complements $\mathfrak{n}_\lambda\mathrm{:=}\mathfrak{h}_\lambda^\bot$, with $\lambda\in\{\xi_\bullet,\eta_\bullet\}$, are taken with respect to an invariant inner product, like the Killing form. Equivariance is guaranteed by the fact that there is a $G_{\xi_{\bullet}}$-equivariant identification between $N_{\xi_\bullet}$ and the dual space $\mathfrak g_{\xi_\bullet}^{*}$. Notice that $\mathfrak{n}_\lambda$ is not in general a subalgebra, but rather just a vector subspace of $\mathfrak{g}$.

In other words, consider the sequences of vector spaces for $\lambda\in\{\xi_\bullet,\eta_\bullet\}$
\begin{equation}\label{seq}
0\rightarrow\mathfrak{h}_\lambda \rightarrow \mathfrak{g}\rightarrow T_\lambda\mathcal{O}_\lambda\simeq\mathfrak{n}_{\lambda}\rightarrow 0;
\end{equation}
if the vertical distribution is pointwise isomorphic to the quotient $\qsp{\mathfrak{h}_{\xi_\bullet}}{\mathfrak{h}_{\eta_\bullet}}$ a choice of embedding of $\mathfrak{n}_{\xi_\bullet}$ in $\mathfrak{n}_{\eta_\bullet}$ will define the horizontal distribution. We can  interpret the inclusion $\mathfrak{n}_{\xi_\bullet}\subset\mathfrak{n}_{\eta_\bullet}$ as an inclusion of tangent spaces: ``$T_{\pi(\eta)}\mathcal{O}_{\xi_\bullet}\subset T_{\eta}\mathcal{O}_{\eta_\bullet}\!\!$" for $\eta\in\orb_{\eta_\bullet}$, and what this truly means is that the splitting $\mathfrak{g}=\mathfrak{h}_{\xi_\bullet}\oplus\mathfrak{n}_{\xi_\bullet}$ { yields} a natural affine connection on those associated bundles given by the orbits $\mathcal{O}_{\eta_\bullet}$ whose stabiliser { $H_{\eta_\bullet}$ is a closed subgroup of $H_{\xi_\bullet}$. As a matter of fact, since the splitting is equivariant \cite{Guill} and can be cast at every point of the total space, it yields} a smooth distribution in $T\mathcal{O}_{\eta_\bullet}$, i.e. a smooth, equivariant assignment of an horizontal subspace of the tangent space at every point. 

So, summarising, the tangent spaces to $\mathcal{O}_{\eta_\bullet}$ and $\mathcal{O}_{\xi_\bullet}$ are isomorphic (as vector spaces) respectively to $\mathfrak{n}_{\eta_\bullet}$ for the total space and to ${\mathfrak{n}_{\xi_\bullet}}$ for the base space, and the inclusion ${\mathfrak{n}_{\xi_\bullet}}\subset\mathfrak{n}_{\eta_\bullet}$ defines the horizontal distribution in $T\orb_{\eta_\bullet}$.

\begin{remark}
All tangent spaces $T_\lambda\orb_{\lambda_\bullet}$ are isomorphic to $\mathfrak{n}_{\lambda_\bullet}$ \emph{via} the adjoint action. For the sake of clarity we shall distinguish the images of the tangent embeddings from one another (cf. Eq. \eqref{natdiff}).
\end{remark}

\begin{remark}Notice that the sequence \eqref{seq} splits, and since $T_\lambda\mathcal{O}_\lambda\simeq\mathfrak{n}_{\lambda}$ we have that there exists a section of the surjective map $\phi\equiv\ad\lambda: \mathfrak{g}\longrightarrow T_\lambda\mathcal{O}_\lambda$. As a matter of fact, if we are given a tangent vector $v_\lambda\in T_\lambda\mathcal{O}_\lambda$ we can obtain its preimage in $\mathfrak{n}_{\lambda}$. Let us call such a map 
\begin{equation}\label{phimap}
\phi^{-1}: T_\lambda\mathcal{O}_\lambda\longrightarrow\mathfrak{n}_{\lambda}.
\end{equation}
More explicitly, if $v_\lambda = \ad_K\lambda$ with $K\in\mathfrak{g}$, we have {  $\phi^{-1}(v_\lambda)=K\big|_{\mathfrak{n}_{\lambda}}$}.
\end{remark}

Consider now the ($G$-equivariant) embedding

\begin{equation}\label{inclusion}
\iota: \orb_{\lambda_\bullet} \longrightarrow \mathfrak{g}
\end{equation}
and its ($G$-equivariant) tangent map
\begin{equation}\label{tangemb}
\iota_*: T\orb_{\lambda_\bullet} \longrightarrow T\mathfrak{g}
\end{equation}

Notice that, since the tangent space to $\orb_{\lambda_\bullet}$ above every point $\lambda$ is identified with $\mathfrak{n}_{\lambda}$, the tangent map will be an isomorphism onto its image $\mathfrak{n}_{\lambda}$. This means that it defines a Lie algebra valued 1-form that pointwise looks like:

\begin{equation}\label{natdiff}
\iota_*\big|_\lambda\equiv d\lambda \in T^*_\lambda\orb_{\lambda_\bullet} \otimes\mathfrak{n}_{\lambda}
\end{equation}
where by $d\lambda$ we \emph{also} mean the covector obtained by applying the DeRham differential to $\lambda$, as it can be expressed in a local chart (see for instance \cite[Eq. 50]{ES1}), although this is the derivative of the embedding map and it does not depend on a local chart.  

In what follows we will always use the notation $d\lambda$ for simplicity; what we mean is that the action of $d\lambda$ on a tangent vector $v_\lambda\in T_\lambda\orb_{\lambda_\bullet}$ yields its \emph{bare or tautological} value, { denoted by} $d\lambda(v_\lambda)$. When $G$ is a group of matrices, say for instance $U(n)$, the value of $d\rho$, with $\rho\in\orb_{\rho_0}$ will be the parametrized matrix of differentials, obtained by differentiating all the entries with respect to the $\mathrm{dim}(G) - \mathrm{dim}(H)$ independent local coordinates.

\begin{remark}
The 1-forms $d\eta$ and $d\xi$ can be seen as 1-forms with values in the ring of functions on the quotient space, following the sequence \eqref{Seq}. This induces a sequence of Lie algebra modules at the level of the tangent spaces, so that quantum states can be seen as fiberwise linear functions on the tangent space.
\end{remark}

\subsection{The fibration for mixed states}\label{maxfibration}
Let us fix $G=U(n)$ in what follows. We will specify the previous discussion for this particular case, in order to connect with the theory of mixed states.

We are ready now to study how to the fibration of co-adjoint orbits is interpreted when dealing with the space of states of quantum mechanics. Let us pick a  generic mixed state $\rho_0$, i.e. a mixed state with $m_{\rho_0}=\{1,\dots,1\}$, and let $\widetilde{\rho}_0$ represent any other mixed state. The theory of co-adjoint orbits for the compact group $U(n)$ applies to the present case and the stabiliser of  generic mixed states like $\rho_0$ is always a maximal torus $T^n$, which is also the minimal stabiliser. This means that the orbits associated with such states are those with maximal dimension, and they are diffeomorphic to to the coset:
\begin{equation}\label{fullflag}
\mathcal{F}^n\simeq\qsp{U(n)}{T^n}
\end{equation}
which is  also referred to as \emph{complete flag manifold}.

We can conclude that the mixed states $\rho_0$ and $\widetilde{\rho}_0$ define a fibration through the associated bundle
\begin{equation}
U(n)\times_{H_{\widetilde{\rho}_0}}\left(\Ad_{H_{\widetilde{\rho}_0}} \rho_0\right)\simeq \mathcal{O}_{\rho_0}
\end{equation}
represented by the diagram:
\begin{align}
\xymatrix{
\qsp{H_{\widetilde{\rho}_0}}{T^n}\ar@{^{(}->}[r] & \orb_{\rho_0}\simeq\qsp{U(n)}{T^n} \ar@{->>}[d] \\
& \orb_{\widetilde{\rho}_0}\simeq\qsp{U(n)}{H_{\widetilde{\rho}_0}}.
}
\end{align}

If the stabiliser of the mixed state is, say, $H_{\widetilde{\rho}_0}=U(m_1)\times\dots\times U(m_N)$, with $N$ being the length of $m$, the fibre will be isomorphic to the orbit 
\begin{equation}\begin{aligned}
\qsp{U(m_1)\times\dots\times U(m_k)}{T^n}&\simeq\qsp{U(m_1)}{T^{m_1}}\times\dots\times\qsp{U(m_N)}{T^{m_N}}\\
&\simeq\bigtimes\limits_{i=1}^N\mathcal{F}^{m_i}
\end{aligned}\end{equation}
where we set $\mathcal{F}^{m_i}\simeq\qsp{U(m_i)}{T^{m_i}}$ to denote the i-th full flag manifold associated with \emph{some subset} $\tau_i$ of the partition of unity $\kappa$. It corresponds to the $m_i$ eigenvalues of $\rho_0$ from the $(l_i=\sum_j^im_j)$-th eigenvalue to the $(l_i+m_i)$-th. Then the fibration reads
\begin{align} \label{general nesting}
\xymatrix{
\bigtimes\limits_{i=1}^N\mathcal{F}^{m_i}\ar@{^{(}->}[r]& \mathcal{F}^n\ar@{->>}[d]\\
& \mathcal{O}_{\widetilde{\rho}_0}.
}
\end{align}

This general property of co-adjoint orbits is interpreted in the context of mixed states as a nesting of spaces of states with different dimensions and mixing coefficients. When $\widetilde{\rho}_0$ is a pure state, its stabiliser (maximal in U(n)) is given by $U(n-1)\times U(1)$ and the orbit under the adjoint action is diffeomorphic to the projective space:
\begin{equation}
\mathbbm{CP}^{n-1}\simeq\qsp{U(n)}{U(n-1)\times U(1)}.
\end{equation}

Furthermore, we have that $T^n\subset U(n-1)\times U(1)$, and if we consider the splittings in the Lie algebra
\begin{equation}
\mathfrak{u}(n)\simeq
\mathfrak{t}\oplus\mathfrak{n}\simeq\widetilde{\mathfrak{k}}\oplus\widetilde{\mathfrak{n}}
\end{equation}
where $\widetilde{\mathfrak{k}}$ is the Lie algebra of the stabiliser $U(n-1)\times U(1)$ and $\widetilde{\mathfrak{n}}$ is its orthogonal complement, we have the vector space inclusions
\begin{equation}\begin{aligned}
\mathfrak{t}&\subset\widetilde{\mathfrak{k}},\\ \mathrm{Im}(\ad{}\widetilde{\rho}_0)\simeq \widetilde{\mathfrak{n}}&\subset\mathfrak{n}\simeq\mathrm{Im}(\ad{}\rho_0)
\end{aligned}\end{equation}
yielding the fibration
\begin{equation}\label{maxminorb}
\mathcal{F}^{n-1} \longrightarrow \mathcal{F}^n \longrightarrow \mathbb{CP}^{n-1}.
\end{equation}

This argument is recursive in $n$, and choosing generic mixed states and pure states of different dimensions we get the sequence of fibrations:
\begin{align} \label{multiple nesting}
\xymatrix{
\mathbbm{CP}^{1}\simeq\mathcal{F}^{2}\ar@{^{(}->}[r]&\mathcal{F}^{3}\ar@{->>}[d]\ar@{^{(}->}[r]&\mathcal{F}^4\ar@{->>}[d]&\dots &\mathcal{F}^{n-1}\ar@{^{(}->}[r] \ar@{->>}[d]  & \mathcal{F}^n \ar@{->>}[d]\\
&\mathbbm{CP}^{2}\ar@{^{(}->}[r]&\mathbbm{CP}^3&\dots&\mathbbm{CP}^{n-2} \ar@{^{(}->}[r]& \mathbbm{CP}^{n-1}.
}
\end{align}

The spaces on top are the spaces of generic mixed states of increasing dimension, and the bases are the projective spaces of pure states. Even though the pure states space $\mathbb{CP}^1$ and the mixed states space $\mathcal{F}^2$ are topologically the same, they are different \emph{as spaces of quantum states}. Indeed they are associated with different representatives in the Lie algebra, and it is worthwhile to distinguish them under this perspective.

\subsubsection{A worked-out example}
To fix the ideas, consider the basic example of $i\mathfrak{u}(3)$, with the generators given by the Gell-Mann matrices $\lambda_i$ plus the identity matrix. The fibration reads
\begin{equation}\label{U(3)}
\mathcal{F}^2 \longrightarrow\mathcal{F}^3\simeq\orb_{\rho_0}\stackrel{\pi}{\longrightarrow}\mathbbm{CP}^2\simeq\orb_{\widetilde{\rho}_0}.
\end{equation}
The \emph{normal} subspaces are generated by
\begin{equation}\begin{aligned}
\mathfrak{n}&=\mathrm{Span}\{\lambda_1,\lambda_2,\lambda_4,\lambda_5,\lambda_6,\lambda_7\}\\
\mathfrak{\widetilde{n}}&=\mathrm{Span}\{\lambda_1,\lambda_2,\lambda_4,\lambda_5\}
\end{aligned}\end{equation}
and the tangent maps, computed above the reference points $\rho_0$ and $\widetilde{\rho}_0$, are expanded as \cite{ES2}:
\begin{equation}\begin{aligned}
d\rho_0&=D_1\lambda_1 + D_2\lambda_2+ D_4\lambda_4+ D_5\lambda_5+ D_6\lambda_6+ D_7\lambda_7\\
d\widetilde{\rho}_0&=D_1\lambda_1 + D_2\lambda_2+ D_4\lambda_4+ D_5\lambda_5
\end{aligned}\end{equation}
where the $D_i$'s are $\R$-valued 1-forms. The values of the 1-forms above all other points are obtained through the (co-)adjoint action in an equivariant fashion.

It is clear from this example how the splitting actually defines the horizontal sub-bundle of the tangent space to be the projection of vectors of $\mathfrak{n}$ onto the subspace $\mathrm{Span}\{\lambda_1,\lambda_2,\lambda_4,\lambda_5\}$. Moreover, the remaining part, namely $\mathrm{Span}\{\lambda_6,\lambda_7\}$ is isomorphic (as a two-dimensional vector space) to $\mathrm{Span}\{\sigma_1,\sigma_2\}$, the off diagonal Pauli matrices that generate the normal complement
\begin{equation}
\mathrm{Span}\{\sigma_1,\sigma_2\}\simeq \qsp{\mathfrak{su}(2)}{\mathfrak{t}_2}\simeq T_x\mathcal{F}^2\simeq\mathrm{ker}(\pi)\simeq \mathrm{Span}\{\lambda_6,\lambda_7\}
\end{equation}
which defines the natural vertical sub-bundle of $T\mathcal{F}^3$ with the projection $\pi$ given in \eqref{U(3)}.

\subsection{The $\mathbb{D}$-map}
The following discussion is important to understand how the Kirillov-Kostant-Souriau (KKS) symplectic form is computed (cf. \eqref{KKS} in Appendix \ref{App:coad}) and to put on the same footing what will be discussed in Section \ref{SYLODE} and subsequent.

Observe that the tangent embedding in \eqref{tangemb}, which is now denoted by $d\rho\equiv \iota_*\big|_\rho:T_\rho\orb_{\rho_0} \longrightarrow \mathfrak{n}_{\rho}$, is not equivalent to the map $\phi^{-1}$ defined in \eqref{phimap}. As a matter of fact, given two tangent vectors $v_\rho,w_\rho\in T_\rho\orb_{\rho_0}$ we can retrieve their preimages $K_v = \phi^{-1}v_\rho$ such that $v_\rho = [K_\rho,\rho]$. 
Moreover, it is possible to promote $\phi^{-1}$ to a $\mathfrak{n}_{\rho}$-valued 1-form $\Phi^{-1}$ defined pointwise by
\[
\Phi^{-1}\big|_\rho = \phi_\rho^{-1} \equiv (\ad\rho)^{-1}{ \big|_{\mathfrak{n}_\rho}}: T_\rho\orb_{\rho_0} \longrightarrow \mathfrak{n}_{\rho}
\]
and we can interpret the object $\Phi^{-1}\wedge \Phi^{-1}$ as a 2-form acting precisely as the KKS-form acts on vector fields. Namely\footnote{Notice the factor of 2 in \eqref{KKSnew} to compensate the antisymmetrisation.}:
\begin{align}\label{KKSclass}
\Omega_{KKS}|_\rho(\ad_K\rho, \ad_H\rho) =& \Tr{\rho[K,H]}\\
\Omega_{KKS}|_\rho(v_\rho,w_\rho) =& 2\Tr{\rho \left(\Phi^{-1}\wedge \Phi^{-1}\right)\big|_\rho(v_\rho,w_\rho)}.\label{KKSnew}
\end{align}

Let us observe that the KKS sympletic form is, by construction, equivariant with respect to the (co)-adjoint action, and this will apply to all geometric tensors we are going to introduce in the following sections.

On the other hand, we can make sense of $v_\rho=\ad_K\rho\equiv[K,\rho]$ as a matrix by (implicitly) applying the map $d\rho$ to obtain an element of $\mathfrak{n}_{\rho}$. It turns out that (on $\rho_0=\mathrm{diag}\{k_i\}$) one may compute:
\begin{align*}
v_{\rho_0}{=}&[K_0,\rho_0] \\ 
(v_{\rho_0})_{ij}{=}&(k_i - k_j) (K_0)_{ij}
\end{align*}
where one should observe that we are implicitly using the isomorphism $d\rho_0:T_{\rho_0}\orb_{\rho_0}\stackrel{\sim}{\longrightarrow}\mathfrak{n}_{\rho_0}$. More explicitly, the equations can be written as:
\begin{align}\notag
d\rho_0(v_{\rho_0})=&[K_0,\rho_0] \\\label{KDRHO}
 (d\rho_0(v_{\rho_0}))_{ij}=&(k_i - k_j) (K_0)_{ij} \\\notag
 =& (k_i - k_j) (\phi^{-1}v_{\rho_0})_{ij}.
\end{align}

Thus, we obtain the diagram:
\begin{equation}
\xymatrix{
\mathfrak{n}_{\rho} \ar@{->}[r]_{\phi_\rho} & \ar@/_1pc/[l]_{\phi_\rho^{-1}}  \ar@/^1pc/[l]^{d\rho} T_\rho\orb_{\rho_0}
} 
\end{equation}
where again $d\rho$ gives the matrix value of a tangent vector $v_\rho$ tautologically and it is not a section of $\phi_\rho\equiv\ad\rho|_{\mathfrak{n}_{\rho}}$, i.e. $\phi_\rho\circ d\rho \not= \mathrm{id}$.
As a matter of fact there exists a map 
\begin{equation}\label{Dmap}
D:\begin{array}{ccc}
\mathfrak{n}_{\rho_0} & \longrightarrow & \mathfrak{n}_{\rho_0} \\
(A)_{ij} & \longmapsto & (k_i - k_j)(A)_{ij}
\end{array}
\end{equation}
such that (cf. Eq. \eqref{KDRHO})
\begin{equation}\label{Ddrho}
d\rho_0(v) = D\circ \phi_{\rho_0}^{-1}(v).
\end{equation}

This map can be {  trivially} extended to a vector bundle morphism covering the identity 
\begin{equation}
\xymatrix{
T\orb_{\rho_0}\otimes\mathfrak{n}_{\rho_0}\ar@{->}^{\mathbb{D}}[r] \ar@{->>}[d]& T\orb_{\rho_0}\otimes\mathfrak{n}_{\rho_0} \ar@{->>}[d] \\
\orb_{\rho_0}\ar@{->}^{\mathrm{Id}}[r]& \orb_{\rho_0}
}
\end{equation}
through the assignment $\mathbb{D}=\mathrm{id}\otimes (-D)$.

\begin{remark}
Recall that the fibre above each point $\rho$ consists of a vector space $\mathfrak{n}_\rho$ that is isomorphic to $\mathfrak{n}_{\rho_0}$ in an equivariant fashion. The map above any point $\rho=U\rho_0U^\dag$ on the $\mathfrak{n}_\rho$ component of the fibre is given by: $\mathbb{D}_\rho=id_\rho \otimes (\Ad_U\circ -D \circ \Ad_{U^{-1}})$. So we will use the shorthand notation, understanding the retraction to the reference point.
\end{remark}

Notice that we defined $\mathbb{D}$ changing the sign to $D$. This is motivated by the following argument, in order to have the right sign on Lie algebra valued 1-forms.

\begin{lemma}\label{skewD}
The map $\mathbb{D}$ is skew symmetric and invertible.
\end{lemma}
\begin{proof}
A simple computation shows that $D$ is skew symmetric with respect to the Killing form in $\mathfrak{u}(n)${ , which we denote by round brackets $(,)$:}
\begin{align*}
(A,D(B))=\Tr{A_{ij} D(B)_{jl}}=&A_{ij}(\rho_j-\rho_i)B_{ji}= \\ 
=&-(\rho_i-\rho_j)A_{ij}B_{ji} = -\Tr{D(A)_{ij}B_{jk}}=(-D(A),B).
\end{align*}
Extending to any inner product on $T\orb_{\rho_0}\otimes\mathfrak{n}_{\rho_0}$, with respect to which $\mathrm{id}$ is trivially symmetric, we can deduce the skew symmetry of $\mathbb{D}$. 

To show the invertibility we have to analyse what happens for $k_i=k_j$. Recall that $D$ is an endomorphism in $\mathfrak{n}_{\rho_0}$; the condition of two eigenvalues of $\rho_0$ being equal means that the stabiliser subalgebra $\mathfrak{h}_{\rho_0}$ is larger, and the corresponding orthogonal complement is smaller. Therefore we obtain that the $(ij)$-component is trivially zero and there is nothing to invert, concluding the proof. 
\end{proof}

Using the bundle morphism and Lemma \ref{skewD} we have the following:
\begin{proposition}\label{intermedpull}
The two form $\Omega_\orb|_\rho:=\Tr{\rho d\rho\wedge d\rho}$ is the pullback of the KKS symplectic form along $\mathbb{D}$:
\begin{equation}
\Omega_\orb = \frac12 \mathbb{D}^*\Omega_{KKS}.
\end{equation}
\begin{proof}
We only need the following identity:
\[d\rho\wedge d\rho = \mathbb{D}^*\left( \Phi^{-1}\wedge \Phi^{-1}\right)
\]
where $d\rho=-\mathbb{D}^*\Phi^{-1}$ has the right sign ($\mathbb{D}=\mathrm{id}\otimes (-D)$) and agrees with Eq. \eqref{Ddrho}. Recalling how we have rewritten the KKS form \eqref{KKSnew}, we can conclude what we claimed.
\end{proof}
\end{proposition}

\begin{remark}
This is to justify the point of view adopted here, which might seem somehow arbitrary and unnecessarily involved.

The role of $\mathbb{D}$ will be relevant for the discussion on the symmetric logarithmic differential, introduced in Section \ref{SYLODE}, and its interplay with the KKS form. Intuitively, the pullback of $\Omega_{KKS}$ under $\mathbb{D}$  is put on the same footing as the pullback of usual differential $d\rho$ to the symmetric logarithmic differential under a map $\mathbb{L}$, of the same nature as $\mathbb{D}$.

Although the construction and computation of the KKS form are clear and well-established in the literature, the previous discussion will turn out useful in the understanding of the rest.
\end{remark}

\section{Symmetric Logarithmic Differential}\label{SYLODE}
In the literature on quantum information theory and quantum metrology, to generalize the idea of \emph{logarithmic derivative} of a probability distributions, various authors have considered and developed the notion of \emph{symmetric logarithmic derivative} for mixed states. In the original definition \cite{Braun, Burb}, it is an Hermitian matrix $\mathsf{L}$ satisfying the implicit equation
\begin{equation}
\partial\rho = \frac{1}{2}\left\{\mathsf{L},\rho\right\}
\end{equation}
where by $\partial\rho$ usually one means the matrix of derivatives of the entries of a mixed state $\rho$ with respect to some external parameter, say, $\theta$ and the curly brackets stand for the \emph{anticommutator} of matrices.

This definition has to be modified if one wishes to stress the geometric features of quantum states in the density matrices formalism. It can be \emph{globalized} by dropping the $\theta$-dependence and by considering the \emph{total differential} $d\rho$ instead of the matrix of derivatives. Doing so we notice that for the equation to make sense it is necessary to require that the hermitian matrix $\mathsf{L}$ be also a well defined Lie algebra valued 1-form. We give then the following definition:

\begin{definition} Given a mixed state $\rho=\Ad_U\rho_0$ we define a \ev{generalized symmetric logarithmic differential} (GSLD) to be an equivariant Lie algebra valued 1-form $\dl\rho\in\left(\Omega^1(\mathcal{O}_{\rho_0})\otimes i\mathfrak{u}(n)\right)^{U(n)}$ satisfying the implicit equation:
\begin{equation}\label{SLD}
d\rho=\frac{1}{2}\left\{\dl\rho,\rho\right\}
\end{equation}
where again the curly brackets stand for the anticommutator of matrices.
\end{definition}

Notice that this definition introduces a uniqueness issue regarding the possible solutions to equation \eqref{SLD} and its homogeneous analogous. A discussion of this issue was carried out first in \cite{ES2}. 

\subsection{Degeneration and uniqueness}

In the previous sections, (see Eq. \eqref{natdiff}),  we saw that $d\rho$ is an equivariant form, i.e. the value of $d\rho$ above any point $\rho=\Ad_U\rho_0$ is obtained through the following equation:
\begin{equation}
d\rho=\Ad_Ud\rho_0
\end{equation}
where $\rho_0$ is the diagonal representative of $\rho$ in the orbit $\orb_{\rho_0}$. For this reason any GSLD must also be equivariant: $\dl\rho=\Ad_U\dl\rho_0$, as we required in the very definition of $\dl\rho$
\begin{equation}
d\rho=\Ad_Ud\rho_0=\frac{1}{2}\Ad_U\left\{\dl\rho_0,\rho_0\right\}=\frac{1}{2}\left\{\Ad_U\dl\rho_0,\rho\right\}.
\end{equation}

This implies that we may always consider $\rho_0$ to be diagonal in $i\mathfrak{u}(n)$ and classify the solutions to \eqref{SLD} according to the eigenvalue multiplicity of the reference point $\rho_0$.

Consider now the matrix equation
\begin{equation}\label{matSLD}
A=\frac12\{X, \rho_0\}
\end{equation}
with $\rho_0\in\mathfrak{t}$ a mixed state as above (a positive definite Hermitian matrix with $\Tr{\rho}=1$). Equation \eqref{matSLD} is well defined in $i\mathfrak{u}(n)$ and we have the following
\begin{lemma}
Whenever $A$ is an element of $\mathfrak{n}_{\rho_0}$, the set of solutions of \eqref{matSLD} is an affine space, which will be denoted by $\mathsf{Log}_{\rho_0}(A)\in i \mathfrak{u}(n)$, and the particular solution is another element of $\mathfrak{n}_{\rho_0}$. This yields a well-defined linear automorphism $\mathsf{L}: \mathfrak{n}_{\rho_0}\longrightarrow \mathfrak{n}_{\rho_0}$.
\end{lemma}
\begin{proof}
The proof of this follows by considering the vector space $V_0$ of solutions of the associated homogeneous equation:
\[
V_0=\{X_H\in i\mathfrak{u}(n)\ \big|\ \{X_H,\rho_0\}=0\}
\] 
and showing, by direct computation (see \cite{ES2}), that there exists a matrix $L_A\in\mathfrak{n}_{\rho_0}$ such that
\[
(L_A)_{ij}=\frac{2}{k_i+k_j}(A)_{ij}
\]
which yields a particular solution. 

It is easy to check that we get a single class of solutions by modding out the space of homogeous solutions, and this is equivalent to intersecting the space of solutions with the normal complement $\mathfrak{n}_{\rho_0}$, namely:
\[
\mathsf{Log}_{\rho_0}(A)\cap \mathfrak{n}_{\rho_0}\simeq\qsp{\mathsf{Log}_{\rho_0}(A)}{V_0}=\{[L_A]\}.
\]
The solution is unique for all $A$'s in $\mathfrak{n}_{\rho_0}$ and we denote it by $\mathsf{L}(A)\in\mathfrak{n}_{\rho_0}$. 
Therefore, we have a well-defined linear automorphism 
\[
\mathsf{L}: \mathfrak{n}_{\rho_0}\longrightarrow \mathfrak{n}_{\rho_0}\ ;\  A\longmapsto \mathsf{L}(A)=L_A.
\]
\end{proof}

We can extend this result and use it to define a vector bundle morphism over the orbit of any mixed state,  as we did for the map $D$, in \eqref{Dmap}:

\begin{lemma}\label{lemma:vbmor}
Given any co-adjoint orbit $\orb_{\rho_0}$ of the Unitary group for a mixed state $\rho_0\in\mathfrak{t}\subset i\mathfrak{u}(n)$, there exists a vector bundle morphism covering the identity
\begin{equation}
\xymatrix{
T\orb_{\rho_0}\otimes\mathfrak{n}_{\rho_0}\ar@{->}^{\mathbb{L}}[r] \ar@{->>}[d]& T\orb_{\rho_0}\otimes\mathfrak{n}_{\rho_0} \ar@{->>}[d] \\
\orb_{\rho_0}\ar@{->}^{\mathrm{Id}}[r]& \orb_{\rho_0}.
}
\end{equation}
Moreover, $\mathbbm{L}$ is symmetric.

\end{lemma}
\begin{proof}
Consider the vector bundle
\[\begin{array}{c}
T\orb_{\rho_0}\otimes\mathfrak{n}_{\rho_0}\\\downarrow\\\orb_{\rho_0}\end{array}
\]
and extend the above automorphism $\mathsf{L}$ to a fibrewise linear action:
\begin{equation}
\mathbb{L}:\begin{array}{c}
T\orb_{\rho_0}\otimes\mathfrak{n}_{\rho_0}\longrightarrow T\orb_{\rho_0}\otimes\mathfrak{n}_{\rho_0}\\
(\rho,v_\rho)\otimes A\longmapsto(\rho,v_\rho)\otimes\mathsf{L}(A).
\end{array}
\end{equation}
Then we have $\mathbb{L}:=\mathrm{Id}\otimes\mathsf{L}$, and we understand the retraction to the reference point through $\Ad_U\circ L\circ Ad_{U^{-1}}:\mathfrak{n}_\rho\rightarrow\mathfrak{n}_\rho$ with $\rho=U\rho_0 U^\dag$.

Using the same computation we used in Lemma \ref{skewD}, we can show that for all $A,B\in\mathfrak{n}_{\rho_0}$, one has $\Tr{AL(B)}=\Tr{L(A)B}$, and extending to any inner product on $T\orb_{\rho_0}\otimes\mathfrak{n}_{\rho_0}$, we get the symmetry of $\mathbb{L}$ as a bundle map.
\end{proof}

We are now ready to prove:

\begin{proposition}[Proposition-Definition]\label{PropSLD}
Let $\rho_0$ represent a mixed state. On the co-adjoint orbit $\orb_{\rho_0}$ of the unitary group $U(n)$, equation \eqref{SLD}:
\[
d\rho=\frac12 \{\dl\rho,\rho\}
\]
has a unique solution $\dl\rho\in\left(\Omega^1(\orb_\rho)\otimes\mathfrak{n}_\rho\right)^{U(n)}$. Such a solution will be called the \ev{symmetric logarithmic differential} (SLD) associated with the mixed state $\rho_0$. Moreover $\dl\rho=\mathbb{L}^*d\rho$.
\end{proposition}
\begin{proof}
Notice that equation \eqref{matSLD} extends to an equivariant equation as
\begin{equation}
\alpha_\rho=\frac12\{\beta_\rho,\rho\}
\end{equation}
everytime $\alpha$ is an equivariant Lie-algebra valued 1-form over $\orb_{\rho_0}$. In particular, when $\alpha= \iota_*\in\left(\Omega^1(\orb_{\rho_0})\otimes\mathfrak{n}_{\rho_0}\right)^{U(n)}$ the equation takes the form
\begin{equation}
d\rho=\frac12\{\dl\rho,\rho\}.
\end{equation}

We know, from the equivariance property of $\dl(\rho),d\rho$ and $\rho$, that solutions to this equation can be recovered by solving the representative equation
\begin{equation}\label{repSLD}
d\rho_0=\frac12\{\dl\rho_0,\rho_0\}.
\end{equation}
As a matter of fact, equation \eqref{repSLD} is of the form \eqref{matSLD}, and has a unique solution $\dl\rho_0\in\mathfrak{n}_{\rho_0}$. In virtue of the above construction we can say that 
\[
\dl\rho=\mathbb{L}^*d\rho
\]
regarded as a $U(n)$-equivariant section of the vector valued cotangent bundle $T^*\orb_{\rho_0}\otimes\mathfrak{n}_{\rho_0}$.
\end{proof}

\section{Fisher tensor}\label{FishT}
First introduced in \cite{ES1,ES2}, to generalize the concept of quantum Fisher information metric, the Fisher tensor is a (0,2)-type tensor that uses the symmetric logarithmic derivative as a fundamental object. Using the construction above we can now provide a precise definition for it and analyse the implications.
\begin{definition}
Let $\orb_{\rho_0}$ be the orbit associated with a mixed state $\rho_0$. We define the \ev{Fisher tensor} $\mathfrak{F}$ to be a section of $\left(T^*\orb_{\rho_0}\right)^{\otimes 2}$ that locally reads
\begin{equation}
\mathfrak{F}_\rho=\Tr{\rho\,\dl\rho{\otimes}\dl\rho}.
\end{equation}
The Fisher tensor can be decomposed in its symmetric and antisymmetric part that will be denoted respectively by $\mathfrak{F}^{\odot}$ and $\mathfrak{F}^\wedge$. 
\end{definition}

We would like to have a closer look at its symmetric and antisymmetric parts and interpret them in terms of intrinsic quantities on the co-adjoint orbits of the unitary group.

\subsection{Antisymmetric part}
We can use the construction outlined above to show that 
\begin{theorem}\label{propantisymm}
The antisymmetric part of the Fisher Tensor $\mathfrak{F}^\wedge$ is a symplectic form.  In particular 
\begin{equation}
\mathfrak{F}^\wedge =\frac12 (\mathbb{DL})^*\Omega_{KKS}.
\end{equation}
\end{theorem}
\begin{proof}
Using Proposition \ref{PropSLD} and Lemma \ref{lemma:vbmor} we know that $\dl\rho=\mathbb{L}^*d\rho$, and thus
\begin{equation}
\mathfrak{F}^\wedge=\mathbb{L}^*\Omega_\orb
\end{equation}
where $\Omega_\orb$ is given by $\Omega_\orb\big|_\rho=\Tr{\rho\,d\rho\wedge d\rho}$. Then, applying Proposition \ref{intermedpull} we get
\[
\mathfrak{F}^\wedge =\frac12 \mathbb{L}^* \mathbb{D}^* \Omega_{KKS}=\frac12 (\mathbb{DL})^*\Omega_{KKS}.
\]

This implies that $\mathfrak{F}^\wedge$ is closed, for $d$ and $(\mathbb{DL})^*$ commute, and the nondegeneracy follows from $\mathbb{DL}$ being a fibrewise linear automorphism.
\end{proof}

This means that the antisymmetric part of the Fisher tensor is \emph{essentially} the KKS form, and that we can always reduce the computation of $\mathfrak{F}^\wedge$ to $\Omega_{KKS}$ as we outlined. Explicitly, using the notation introduced in the previous section, we have:
\begin{equation}
\mathfrak{F}^\wedge(v_\rho,w_\rho)|_\rho= (\mathbb{DL})^*\Tr{\rho \left(\Phi^{-1}\wedge\Phi^{-1}\right)(v_\rho,w_\rho)\}}=\frac12\Tr{\rho [L\circ D(K_v),L\circ D(K_w)]}.
\end{equation}

\begin{remark}
Following \cite{ercolessi} we know that the KKS symplectic form is strongly related to the (geometric) Berry phase, and \emph{a fortiori} we can conclude that the Fisher tensor is related with the curvature of the Berry phase connection as well.

This is a nontrivial statement for several reasons. First of all, because a symplectic form on a $(d>2)$-dimensional co-adjoint orbit need not be related to the KKS symplectic form. This fact depends on the second cohomology groups of the co-adjoint orbits of $U(n)$, which are not trivial, and are generally of a dimension higher than $1$, the case of $\mathbb{CP}^n$.

Furthermore this is telling us that the Fisher tensor is a non-trivial intrinsic object that contains relevant information on the space of states for generic, mixed quantum states, and that it should be highly regarded among other similar information indices.
\end{remark}

\subsection{Symmetric part and metrics}
Following what we have done for $\mathfrak{F}^\wedge$, we can regard $\mathfrak{F}^\odot$ as the pullback of an object similar to the KKS symplectic form, namely $\mathfrak{F}^\odot=(\mathbb{DL})^*G$ with
\begin{equation}
G(v_\rho,w_\rho)|_\rho=\Tr{\rho \left(\Phi^{-1}\odot\Phi^{-1}\right)(v_\rho,w_\rho)\}}=\frac12\Tr{\rho \{K_v,K_w\}}
\end{equation}
where $\{K_v,K_w\}$ is the anti-commutator of matrices in $\mathfrak{n}_{\rho_0}$ and $\odot$ is the symmetric tensor product in $T^*\orb_{\rho_0}$.

Nonetheless, as already pointed out in \cite{ES1}, this is not equivalent to the compatible metric that can be computed from the symplectic structure $\omega_{KKS}$ and the complex structure $J$ on co-adjoint orbits via the formula
\[
G_{KKS}(\cdot,\cdot)=\Omega_{KKS}(\cdot, J\cdot)
\]
making it a K\"ahler manifold. It is known \cite{grab} that this compatible triple $(G_{KKS},\Omega_{KKS},J)$ can be obtained via reduction on each orbit of the corresponding Poisson and Riemannian structures defined by the Hermitian product of the original Hilbert space of states. These non equivalent metrics occur already in the two dimensional case of generic orbits in $\mathfrak{u}(2)$, where the complex structure is given simply by the map $\sigma_1\rightarrow \sigma_2$, $\sigma_2\rightarrow - \sigma_1$. However, in the projective (pure state) cases everything reduces, up to a scaling factor, to the Fubini-Study metric, and we cannot distinguish anymore $G_{KKS}$ from $\mathfrak{F}^\odot$ (up to scale). This can be interpreted as saying that in flag manifolds that are more general than $\mathbb{CP}^n$ we have at least two inequivalent ways of generalising the Fubini-Study metric. One is given by the natural KKS form and its compatible metric (provided that we are given an explicit complex structure), whereas the other one is given by the quantum Fisher metric, i.e. the symmetric part of the Fisher tensor.

Observe that the Fisher information index, of central importance in quantum metrology and quantum information theory, is given by the restriction of the symmetric part of the Fisher tensor $\mathfrak{F}^\odot$ on the curve (resp. surface) $\rho(\theta)$, where $\theta$ is some parameter of interest (or set of parameters) \cite{ES1,ES2}. This suggests that we understand better the role of the Fisher tensor and its symmetric part, which is also related to the Bures metric $g_B$ (cf. \cite{akh} and references therein) via
\begin{equation}
g_B:=\frac14\Tr{d\rho\dl\rho}=\frac18\Tr{\rho\dl\rho\odot\dl\rho}=\frac14\mathfrak{F}^\odot
\end{equation}
where we used the defining equation \eqref{SLD} for $\dl\rho$.

In \cite{akh}, the author used a similar approach to the one presented here to compute the Bures metric on the same coset spaces and even allowing the eigenvalues to change (a similar analysis, for different purposes was carried out in \cite[Sections 2.2, 2.3]{ES1}). His result matches with our language in that the tangent part of the Bures metric (i.e. tangent to the orbit), and hence of the Fisher metric is obtained as
\begin{equation}\label{fishbures}
4g_B=\mathfrak{F}^\odot=(\mathbb{D}\mathbb{L})^*\Tr{\rho(\Phi^{-1}\odot\Phi^{-1})}.
\end{equation}  

As a matter of fact, in  \cite{akh}, Eq.ns 16-19, it is shown how to compute the Bures/Fisher metric also when we allow the eigenvalues of $\rho$ to change, still remaining in a single stratum. In this case, one has:
\begin{equation}\label{fishbures2}
g_B=\sum_i\frac{1}{4}\left(\frac{dk_i}{k_i}\right)^2 + \sum_{i<j}\Lambda_{ij}|(U^\dag dU)_{ij}|^2
\end{equation}  
with $\Lambda_{ij}=\frac{(k_i-k_j)^2}{k_i+k_j}$ and $U^\dag dU$ is the Maurer-Cartan form, with values in the Lie algebra. The first term of (\ref{fishbures2}) is clearly the standard logarithmic derivatives of the abelian part of $\rho$ (i.e. of its eigenvalues) while, from the identity $d\rho_0 = [U^\dag dU,\rho_0]$, we understand that the image of $U^\dag dU$ agrees with $\Phi^{-1}$ in $\mathfrak{n}_{\rho_0}$. It is then a matter of a simple calculation to show that \eqref{fishbures} holds.

Another similar analysis of the metric one can endow the space of mixed states with is carried out in \cite{DCB}, where the construction of the information metric goes through the square root map on pure states. There, it is shown that there exists some duality between the space of pure states (rays in a Hilbert space) and the space of density matrices, given by the computation of a probability distribution from the joint data of a ray and a density matrix. This duality is then exploited to prove important results in quantum metrology.

In our framework we can understand this duality as encoding the identification of $\mathfrak{u}(n)$ and its dual, in that every ray $x$ in a Hilbert space is naturally associated with a density matrix by assigning its projector $P_x$. The probability distribution then arises by taking the trace of any density matrix $\rho$ against the projector.

As an explicit non-trivial example, one can look at the case of $U(3)$.  In \cite{ES2} the symmetric logarithmic derivative and the Fisher tensor for a generic mixed state with $\rho_0 = diag(k_1,k_2,k_3)$, and $k_i\neq k_j$, was computed explicitly. This corresponds to the orbit $U(3)/U(1)^3$. Using local charts for which a generic element of $U(3)$ can be written as the $3\times 3$ matrix:
\[ \exp i \left( \begin{array}{ccc} \lambda_1 & z_1 & z_2 \\ z_1^* &  \lambda_2 & z_3 \\ z_2^* & z_3^* &\lambda_3 \end{array} \right) \]
with $\lambda_j \in \mathbb R$ and $z_j \in \mathbb C$, the Fisher tensor reads:
\begin{equation}\begin{aligned}\label{Fish3,3}
\mathfrak{F}_{U(3)}=4\frac{(k_1-k_2)^2}{(k_1+k_2)^2}&\{(k_1+k_2)dz_1\odot dz_1^* -i(k_1-k_2)dz_1\wedge dz_1^*\}\\
+4\frac{(k_1-k_3)^2}{(k_1+k_3)^2}&\{(k_1+k_3)dz_2\odot dz_2^* -i(k_1-k_3)dz_2\wedge dz_2^*\}\\
+4\frac{(k_2-k_3)^2}{(k_2+k_3)^2}&\{(k_2+k_3)dz_3\odot dz_3^* -i(k_2-k_3)dz_3\wedge dz_3^*\}\\
\end{aligned}\end{equation}
where we have defined the \emph{complex} (anti)-symmetrised tensor products as $dz_i\odot dz_i^*=\frac{1}{2}(dz_i\otimes dz_i^* + dz_i^*\otimes dz_i)$ and $dz_i\wedge dz_i^*=\frac{1}{2i}(dz_i\otimes dz_i^* - dz_i^*\otimes dz_i)$.

Notice that, from the above expression, it can be clearly seen that the correct $U(2)$-like expression is recovered when approaching degenerate cases having, e.g., $k_3 =0$ or $k_1=k_2$.

It is a matter of straightforward computations to check that the antisymmetric part of \eqref{Fish3,3} coincides with the pullback of the KKS symplectic form. Observe, furthermore, that the coordinate chart we used agrees with those of \cite{grab}, Section 5.

\section{Outlook and further research} 
In this paper we have reviewed the relevant basics of the orbit theory and how they can be used to give a consistent formulation of finite dimensional quantum mechanics in the density matrix formalism. This turned out to be fundamental in the interpretation of both the symmetric logarithmic derivative (or differential) and the quantum Fisher information index (tensor) from a geometric point of view.

The Fisher tensor defines a symplectic form on co-adjoint orbits which is the pullback of the natural KKS form. Although in the simple case of n-dimensional pure states and two dimensional mixed states it turns out that $\mathfrak{F}^\wedge \propto \Omega_{KKS}$ \cite{ES1} (as one would expect from the fact \cite{God} that $H^2(\mathbb{CP}^n,\R)\simeq\R$), in all other cases this need not be true in general. In the case of flag manifolds, in fact, one has \cite{Reeder} that $H^2(\mathcal{F}^n,\R)\simeq \mathfrak{t}_{n-1}^*\simeq \R^{n-1}$. It is a nontrivial fact that $\mathfrak{F}^\wedge$ and $\Omega_{KKS}$ are related to one another via the vector bundle morphism $\mathbb{DL}$. This opens up several research directions one can explore, from the compatibillity of the Fisher tensor with the fibration of mixed states, to the relation to geometric quantization and K\"ahler geometry.

It is interesting to understand how the Fisher Tensor behaves with respect to the fibration of the orbits associated with the respective mixed states. We can consider, for instance, the nesting
\begin{equation}
\xymatrix{
\mathbbm{CP}^{1}\simeq\mathcal{F}^{2}\ar@{^{(}->}[r]&\mathcal{F}^{3}\ar@{->>}[d]&\dots &\mathcal{F}^{n-1}\ar@{^{(}->}[r] \ar@{->>}[d]  & \mathcal{F}^n \ar@{->>}[d]\\
&\mathbbm{CP}^{2}&\dots&\mathbbm{CP}^{n-2} & \mathbbm{CP}^{n-1}
}
\end{equation}
and we can add a layer to each step:
\begin{equation}
\xymatrix{
T\mathcal{F}^n\otimes\mathfrak{n}_{\rho_0} \ar@{->}[rr]_{\mathbb{L}_{\rho_0}}\ar@{->>}[rd] \ar@{->}[dd]_{\hat{\pi}_*} & &T\mathcal{F}^n\otimes\mathfrak{n}_{\rho_0} \ar@{->>}[ld] \ar@{->}[dd]^{\hat{\pi}_*}\\
&\mathcal{F}^n \ar@{->>}[dd] &\\
T\mathbb{CP}^{n-1}\otimes\mathfrak{n}_{\xi_0} \ar@{->}[rr]|\hole_{\phantom{\ \ \ \ \ \ }\mathbb{L}_{\xi_0}}\ar@{->>}[rd]& & T\mathbb{CP}^{n-1}\otimes\mathfrak{n}_{\xi_0}\ar@{->>}[ld]\\
& \mathbb{CP}^{n-1} &
}
\end{equation}
where $\rho_0$ and $\xi_0$ are the reference points generating respectively the orbits $\mathcal{F}^n$ and $\mathbb{CP}^{n-1}$. The maps $\hat{\pi}_*$ are given by the tangent maps $\pi_*$ tensored with the natural projection coming from the vector space inclusion $\mathfrak{n}_{\xi_0}\rightarrow \mathfrak{n}_{\rho_0}$. Then, in the dual picture for equivariant sections we have:
\begin{equation}
\xymatrix{
\left(\Omega^1(\mathcal{F}^n)\otimes\mathfrak{n}_{\rho_0}\right)^{U(n)}   & & \ar@{->}[ll]_{\mathbb{L}^*_{\rho_0}}\left(\Omega^1(\mathcal{F}^n)\otimes\mathfrak{n}_{\rho_0}\right)^{U(n)}  \\
& \ar@{->}[lu]^{\dl\rho} \mathcal{F}^n \ar@{->>}[dd] \ar@{->}[ru]_{d\rho}&\\
\left(\Omega^1(\mathbb{CP}^{n-1})\otimes\mathfrak{n}_{\xi_0}\right)^{U(n)} \ar@{->}[uu]_{\hat{\pi}^*}  & & \ar@{->}[ll]|\hole^{\phantom{\ \ \ \ \ \ \ }\mathbb{L}^*_{\xi_0}} \left(\Omega^1(\mathbb{CP}^{n-1})\otimes\mathfrak{n}_{\xi_0}\right)^{U(n)} \ar@{->}[uu]_{\hat{\pi}^*} \\
&  \ar@{->}[lu]^{\dl\xi} \mathbb{CP}^{n-1} \ar@{->}[ru]_{d\xi}&
}
\end{equation}
where $d\rho\equiv\iota_*\big|_{\rho}$ and similarly for $d\xi$. 

As expounded in \cite{Guill}, this fibration comes naturally endowed with a symplectic connection and the KKS form is naturally associated with this symplectic fibration. What happens to the symplectic connection when one looks at the symplectic form defined by $\mathfrak{F}^\wedge$ is to be understood. 

On the other hand, it has been shown in \cite{Skerrit} that co-adjoint orbits admit a natural polarization in their complexified tangent bundle, once again related to the natural KKS symplectic form. The question whether our construction provides a new (possibly inequivalent) polarization, arises naturally offering possibly deep insights on the geometric quantisation of generic co-adjoint orbits of the unitary group.

All of these questions are ultimately related to the K\"ahler geometry of co-adjoint orbits \cite{Kir, Picken}. It is not clear yet whether the Fisher tensor gives also a (possibly inequivalent) K\"ahler structure on the spaces of mixed states. Moreover, it would be interesting to understand how the different K\"ahler structures on the orbits in a fibration are related.

Another possibly unrelated area of further investigation is how one could unify this approach with the one presented in \cite{Christandl} for entanglement classes. We strongly believe that the two points of view can be related to one another, once the appropriate modifications (like passing from $(S)U(n)$ to $SL(n,\mathbb{C})$) are taken into account.

All these questions will be subject of further research by the authors.

\newpage
\appendix
\section{Co-Adjoint orbits}\label{App:coad}
In this appendix we will give the basic theoretical notions about co-adjoint orbits of Lie groups. An extensive treatment of the subject can be found in \cite{Kir}, and we will refer to this for all the proofs.

Let us consider a Lie group $G$ and its Lie algebra $\mathfrak{g}$, that is its tangent space $\T{G}{e}$ to the identity $e\in G$. The internal operation in $\mathfrak{g}$ is given by the Lie bracket $[\cdot,\cdot]$ coming from the commutator bracket among left-invariant vector fields evaluated at the identity. As it is known, $G$ acts on itself transitively by left or right multiplication, that is to say, there exist two diffeomorphisms $G\longrightarrow G$:
\begin{equation}
\begin{aligned}
L_g:\ h&\longmapsto L_g h=gh\\
R_{g}:\ h&\longmapsto R_{g}h=hg^{-1}
\end{aligned}
\end{equation}
from which we may construct an automorphism of the group, called conjugation as
\begin{equation}
\mathsf{Cj}_g=L_g\circ R_{g}:\ h\longmapsto ghg^{-1}
\end{equation}
which leaves the identity fixed. In virtue of this property we understand that $\mathsf{Cj}$ naturally defines a linear map $\mathfrak{g}\longrightarrow\mathfrak{g}$:
\begin{equation}
\label{adjointrep}
\Ad_g:\ (\mathsf{Cj}_{g})_*|_e.
\end{equation}

\begin{definition}
Notice that the map $g\longmapsto\Ad_g$ defines a group homomorphism between the group G and the group of automorphisms $\mathsf{Aut}\,\mathfrak{g}$ of the Lie algebra, regarded as a vector space. 

This mapping is called \ev{adjoint representation} of the Lie group G on $\mathfrak{g}$. Since this representation is smooth, we may differentiate it and obtain the mapping 
\begin{equation}
\Ad_{g*}(e)\equiv\ad:\begin{array}{rcl}
\mathfrak{g}&\longrightarrow&\mathsf{End}\,\mathfrak{g}\\
X&\longmapsto&\mathsf{ad}_X
\end{array}
\end{equation}
such that $\ad_XY=[X,Y]$ for $X,Y\in\mathfrak{g}$.
\end{definition}

Similarly, we may now consider $\mathfrak{g}^*$, the dual of $\mathfrak{g}$, and we have the following

\begin{definition}
We define a linear representation of $G$ on $\mathfrak{g}^*$ by means of the adjoint representation as the mapping $g\longmapsto\Cad_g$ via the relation $\Cad_g=(\Ad_{g})^\vee$, where the symbol ${}^\vee$ denotes the dual operator to $\Ad_{g}$ and its action is specified through the pairing:
\begin{equation}
<\Cad_g\xi,X>:=<(\Ad_{g})^\vee\xi,X>=<\xi,\Ad_{g^{-1}}X>.
\end{equation}

This representation is called \ev{co-adjoint representation} and it is linear, therefore it can be derived on the identity to yield a map
\begin{equation}
(\Cad_{g})_*(e):
\begin{array}{rcl}
\mathfrak{g} & \longrightarrow & \mathsf{End}\,\mathfrak{g}^*\\
X & \longmapsto & \ad_X^*
\end{array}
\end{equation}
such that $\ad_X^*$ acts as the adjoint of $\ad_X$, that is to say
\begin{equation}
<\ad_X^*\xi,Y>=<\xi,-\ad_XY>=<\xi,-[X,Y]>
\end{equation}
for every $X,Y\in\mathfrak{g}$ and $\xi\in\mathfrak{g}^*$.
\end{definition}

Let us consider now a distinguished point $\xi_\bullet\in\mathfrak{g}^*$. The bullet denotes the representative role of $\xi_\bullet$ in considering the class of all linear functionals in $\mathfrak{g}^*$ obtained by acting on $\xi_\bullet$ through the co-adjoint-representation of the group $G$. As a matter of fact we are dealing with the orbit $\orb_{\xi_\bullet}\subset\mathfrak{g}^*$, $\orb_{\xi_\bullet}=\{\Cad_g\,\xi_\bullet\ |\ g\in G\}$. 

Consider the subgroup $H_{\xi_\bullet}$of $G$ that leaves $\xi_\bullet$ fixed, that is to say $H_{\xi_\bullet}=\{h\in G\ |\ \Cad_{h}\xi_\bullet=\xi_\bullet\}$. This subgroup is called \ev{stabiliser} subgroup of (the point) $\xi_\bullet$. We have then
\begin{equation}\label{Pbdl}
\orb_{\xi_\bullet}\simeq \qsp{G}{H_{\xi_\bullet}}.
\end{equation}
We denote $\mathfrak{h}_{\xi_\bullet}$ the Lie algebra of the stabiliser of $\xi_\bullet$. 

The action of $G$ on $\orb_{\xi_\bullet}\subset\mathfrak{g}^*$ is transitive and therefore it is possible to view the group $G$ as a fibre bundle $H_{\xi_\bullet}\rightarrow G \rightarrow \orb_{\xi_\bullet}$, where the natural projection map is given by
\begin{equation}
\label{cadfibration}
\pi:\begin{array}{rcl}
G&\longrightarrow&\orb_{\xi_\bullet}\\
g&\longmapsto&\Cad_g\xi_\bullet
\end{array}
\end{equation} 
with fibre above $\xi_\bullet$ isomorphic to $H_{\xi_\bullet}$. By taking the tangent maps we obtain the exact sequence:
\begin{equation}
\label{sequence}
0\rightarrow\mathfrak{h}_{\xi_\bullet}\stackrel{\iota}{\hookrightarrow}\mathfrak{g}\stackrel{d\pi}{\rightarrow}\T{\orb_{\xi_\bullet}}{\xi_\bullet}\rightarrow0
\end{equation}
where the map from $\mathfrak{g}$ to $\T{\orb_{\xi_\bullet}}{\xi_\bullet}$ is given by the derivative $d\pi \equiv \cad$ of the projection and $\iota\colon\ \mathfrak{h}_{\xi_\bullet}\hookrightarrow\mathfrak{g}$ is the inclusion of the stabilizing sub-algebra into $\mathfrak{g}$.  Note that this means that $\mathrm{ker}(\cad \xi_\bullet)=\mathfrak{h}_{\xi_\bullet}$.

\begin{definition}
Let us consider the bilinear form on $\mathfrak{g}$
\begin{equation}
\label{KKS}
\beta_{\xi_\bullet}(X,Y)=\left<\xi_\bullet,[X,Y]\right>=-\left<\cad_X\xi_\bullet,Y\right>.
\end{equation}
Any time $\orb_{\xi_\bullet}$ is a co-adjoint orbit of G in $\mathfrak{g}^*$, there exists a 2-form $\omega_{\xi_\bullet}$ acting pointwise on two vectors of $\T{\orb_{\xi_\bullet}}{\xi_\bullet}$ coming from $\mathfrak{g}$ via the map in \eqref{sequence}. 

$\omega_{\xi_\bullet}$ will be called called KKS 2-form after Kostant, Kirillov and Souriau and has the local expression 
\begin{equation}
\omega_{\xi_\bullet}(\ad_X^*\xi_\bullet,\ad_Y^*\xi_\bullet)=\beta_{\xi_\bullet}(X,Y)
\end{equation}
on the base point $\xi_\bullet$  of the orbit $\orb_{\xi_\bullet}$.
\end{definition}

By virtue of the above construction we may prove the following proposition (see \cite{Kir})
\begin{proposition}
The kernel of $\beta_{\xi_\bullet}$ is exactly the stabilizing algebra $\mathfrak{h}_{\xi_\bullet}$, and the form is invariant under the action of the stabiliser $H_{\xi_\bullet}$. This means that $\omega_{\xi_\bullet}$ is non-degenerate on the tangent $\T{\orb_{\xi_\bullet}}{\xi_\bullet}$ for every $\xi_\bullet$.

\end{proposition}

\begin{remark}
Since $\xi_\bullet$ is only a representative of its orbit, we may act on it with an element $g$ of $G$. The change in $\xi=\Cad_g\xi_\bullet$ will be compensated by the change of the preimages $X_\xi,Y_\xi$:
\[
<\xi,[X_\xi,Y_\xi]>=<\Cad_g\xi_\bullet,\Ad_g[X_\bullet,Y_\bullet]>=<\xi_\bullet,[X_\bullet,Y_\bullet]>.
\]
This can be summarized by saying that the 2-form is $G$-equivariant.
\end{remark}

The most relevant result in this section is given by the following:

\begin{theorem}
On every co-adjoint orbit $\orb$ of a Lie group $G$ there exists a non-degenerate, closed, G-invariant 2-form $\omega$. The pair $(\orb,\omega)$ is thus a symplectic manifold.
\end{theorem}

The theorem is proved by showing that the KKS form we have just constructed is symplectic on $\orb_{\rho_0}$. This can be done by direct calculations, but there exist other, more elegant proofs as discussed in \cite{Kir}.

\section{Compact Lie algebras and orbits}\label{App:generalities}
This appendix is a brief review on Lie algebras of Lie groups that are compact as topological spaces. Results on compact Lie algebras will play a fundamental role in the application to Quantum Mechanics, for we are interested in the particular case of $U(n)$, the group of unitary matrices.

The facts presented here have no ambition of being exhaustive: this is rather a transversal survey of essential results that will be important in the applications. For a more thorough exposition on this subject we refer to \cite{Kir,Hsiang}.

To begin with, we will need the basic definitions that will be used throughout the section:
\begin{definition}
A \ev{semisimple Lie algebra} is a Lie algebra that decomposes in direct sums of \ev{simple} Lie algebras, which in turn are Lie algebras with no non-trivial ideals. In other words a semisimple Lie algebra does not have abelian ideals. A Lie algebra is called \ev{compact} if it is the Lie algebra of a compact Lie group. 

Moreover, we define a symmetric bilinear form on Lie algebras through
\begin{equation}
K(X,Y)=\Tr{\ad X \ad Y}
\end{equation}
such a bilinear is called \ev{Killing-Cartan form} or simply Killing form.
\end{definition}

The first structural theorem for semisimple algebras is the following, due to Cartan:
\begin{theorem}
A Lie algebra $\mathfrak{g}$ is \ev{semisimple} if and anly if the killing form is nondegenerate on $\mathfrak{g}.$
\end{theorem}

By general results in Lie theory we have that a compact Lie algebra with no center $\mathfrak{g}$  can always be considered as the Lie algebra of a connected and simply connected  compact Lie group. Moreover, 
\begin{proposition}\label{structcompact}
Let $\mathfrak{g}$ be a compact Lie algebra, then the Killing form $K(\cdot,\cdot)$ is negative semidefinite and the algebra decomposes as 
\begin{equation}
\mathfrak{g}=\mathfrak{z}(\mathfrak{g}) \oplus [\mathfrak{g},\mathfrak{g}]
\end{equation}
with the second factor being semisimple.
\end{proposition} 
Then we have that $\mathfrak{g}$ compact with no center implies $\mathfrak{g}$ semisimple. This will be fundamental in what follows. We will now look at complex semisimple Lie algebras for a little while.

\begin{definition}
Consider a complex semisimple Lie algebra $\mathfrak{c}$. 

Any endomorphism $E\in\mathrm{End}(V)$ is called \ev{semisimple} if it is diagonalizable. Moreover, we call a subalgebra $\mathfrak{h}\subset\mathfrak{c}$ a \ev{Cartan subalgebra}  if it is maximal abelian and if, everytime $H\in\mathfrak{h}$, its adjoint operator $\ad_H$ is semisimple as an endomorphism.
\end{definition}
About the existence of Cartan subalgebras we have the following classical theorem:
\begin{theorem}
Let $\mathfrak{c}$ be a complex compact semisimple Lie algebra, then there exists a unique Cartan subalgebra $\mathfrak{h}$ up to isomorphisms, which is selfnormalizing, that is to say $N_{\mathfrak{c}}(\mathfrak{h}):=\{\xi\in\mathfrak{c}\ |\ [\xi,\eta]\in\mathfrak{h},\forall \eta\in\mathfrak{h}\}=\mathfrak{h}$.
\end{theorem}

Even if we will be interested in real compact real Lie algebras it is important to have a look at this more general setting, where we consider complex Lie algebras. As we will see, it is in the context of compex semisimple Lie algebras that the following important concepts will find their natural home. 

\begin{definition}
Let $\mathfrak{h}$ be a Cartan subalgebra of the complex semisimple Lie algebra $\mathfrak{c}$, $\alpha\in\mathfrak{h}^*$ and consider the sets 
\begin{equation}
\mathfrak{c}_\alpha:=\{X\in\mathfrak{c}\ |\ [H,X]\equiv \ad_HX=\alpha(H)X,\ \forall H\in \mathfrak{h}\}.
\end{equation}
Everytime the set $\mathfrak{c}_\alpha$ is nontrivial and $\alpha\not=0$ we say that $\alpha\in\mathfrak{h}^*$ is a \ev{root} and the corresponding $\mathfrak{c}_\alpha$ is called \ev{root space} for $\alpha$. The choice $\alpha=0$ gives us $\mathfrak{c}_0\equiv\mathfrak{h}$. 

Denote by $R\subset\mathfrak{h}^*$ the set of all roots, such a set is called \ev{root system}.
\end{definition}
The root spaces $\mathfrak{c}_\alpha$ are roughly speaking the simultaneous eigenspaces for all the operators $\ad_H$ with eigenvalues $\alpha(H)$, moreover the root system $R$ inherits an inner product from $\mathfrak{g}$, namely the dual of the restriction of the Killing form to $\mathfrak{h}$. The span of $R$ is an Euclidean vector subspace $E$ of $\mathfrak{h}^*$, and one can show \cite{Kir} that $K^*(\cdot,\cdot)$ is positive definite on $E$ and that $R$ satisfies the axioms of \emph{nondegenerate, reduced root systems} in $E$.
Then we have the following properties:
\begin{proposition}
Let $R$ be the root system associated to a complex semisimple Lie algebra $\mathfrak{c}$. The following decomposition holds:
\begin{equation}
\mathfrak{c}=\mathfrak{h}\oplus\bigoplus_{\alpha\in R}\mathfrak{c}_\alpha.
\end{equation}
Moreover\
\begin{equation}\label{mult}
[\mathfrak{c}_\alpha,\mathfrak{c}_\beta]\subset\mathfrak{c}_{\alpha+\beta},\ \forall\alpha,\beta\in R.
\end{equation}
The set of roots spans $\mathfrak{h}^*$ and the Killing form vanishes on vectors in different eigenspaces, that is to say $\alpha\not=-\beta$ implies $K(\mathfrak{c}_\alpha,\mathfrak{c}_\beta)=0$. Last, the Killing form is nondegenerate on the Cartan subalgebra. 
\end{proposition}

\begin{remark}
Notice that only $\mathfrak{c}_{-\alpha}$ can have nonzero Killing form with $\mathfrak{c}_\alpha$. Opposite roots are then paired in $R$, and the pairing is indeed given by the Killing form.

Moreover, the nondegeneracy of $K\big|_\mathfrak{h}$ lets us identify $\mathfrak{h}\simeq\mathfrak{h}^*$. Therefore we can denote by $H_\alpha\in\mathfrak{h}$ the dual vector to the root $\alpha$, $K(H_\alpha,H)=\alpha(H)$. It will be called \ev{dual root}. Non zero elements in $\mathfrak{c}_\alpha$ will be denoted by $X_\alpha$ and will be called \ev{root vectors}.
\end{remark}

The following properties also hold:
\begin{proposition}
Let $R$ be a root system for $\mathfrak{h}\subset\mathfrak{c}$ as before. Then
\begin{enumerate}
\item $[X_\alpha,X_{-\alpha}]=K(X_\alpha,X_{-\alpha})H_\alpha$. 

It is always possible to normalise $K(X_\alpha,X_{-\alpha})=1$. and $\mathfrak{h}\equiv\mathfrak{c}_0$ is the complex span of $H_\alpha,\ \alpha\in R$. 
\item $K^*(\alpha,\alpha)\not=0$.
\item $\mathrm{dim}\mathfrak{c}_\alpha=1\ \forall \alpha\in R$.
\end{enumerate}
\end{proposition}

Now we will explain in which sense the information contained in the Cartan subalgebra is redundant and how we can gather all of the relevant information in some conic sector $R^+\subset \mathfrak{h}^*$.

\begin{definition}
A subset $R^+ \subset R$ is said to be a \ev{set of  positive roots} iff for $\alpha\in R^+$ there exists a vector $v\in E$ such that $K^*(\alpha,v)>0$ and $K^*(\beta,v)<0$ for $\beta\in R\backslash R^+$.

The \ev{Weyl group} associated with the root system $R$ is the group generated by reflections through the hyperplanes orthogonal to the roots. As such it is a subgroup of the isometry group of $R$.

We define the \ev{positive Weyl chamber} associated to $R^+$ to be the closed subset
\begin{equation}
\mathcal{C}^+=\{v\in E\ |\ K^*(\alpha,v)\geq 0,\ \alpha\in R^+\}.
\end{equation}

The Weyl group $W$ acts on the positive Weyl chamber and we call any image under any element $w\in W$ a \ev{Weyl chamber}.
\end{definition}

\begin{remark}

Referring to \cite{Kir} for more details, we claim that one can consider the subspace $\mathfrak{h}_\R:=\text{Span}_\R\{iH_\alpha\} \subset \mathfrak{h}$, the real span of the dual roots multiplied by $i$, to obtain what is called a \ev{real form} for $\mathfrak{h}$, that is to say $\mathfrak{h}=\mathfrak{h}_\R\otimes \mathbbm{C}$ and $K$ is positive definite on it. 

Considering the root vectors $X_\alpha, X_{-\alpha}$ and their commutator $H_\alpha=[X_\alpha,X_{-\alpha}]$ for the complex (or complexified) Lie algebra $\mathfrak{c}$, the real span of the combinations
\begin{equation}
(X_{\alpha}-X_{-\alpha}),\ iH_\alpha,\ i(X_{\alpha}+X_{-\alpha})
\end{equation}
defines a real form for $\mathfrak{g}$, called \ev{compact real form}. It is a compact semisimple Lie algebra, whose maximal abelian subgroup is given by the real span of the dual roots $iH_\alpha$.

If we take a compact real semisimple Lie algebra $\mathfrak{g}$ with a maximal abelian subalgebra $\mathfrak{t}$, complexifying we have that $\mathfrak{h}=\mathfrak{t}_\mathbbm{C}$ is a Cartan subalgebra for $\mathfrak{g}_\mathbbm{C}$. Roughly speaking, the subalgebra $\mathfrak{t}$ is almost a Cartan subalgebra of $\mathfrak{g}$: it is maximally abelian and all his elements are ad-skewsymmetric. On the other hand the subalgebra $i\mathfrak{t}\subset \mathfrak{g}_\mathbb{C}$ is isomorphic to $\mathfrak{h}_\R$ and the restriction $K(\cdot,\cdot)|_{i\mathfrak{t}}$ is positive definite.

As a matter of fact, the eigenvalues of $\ad_{H}$ with $H\in\mathfrak{t}$ (i.e. the dual roots of $\mathfrak{g}_\mathbb{C}$) are purely imaginary due to the skew-symmetry and reality of $\ad_H$. If $\alpha$ is a root of $\mathfrak{g}_\mathbb{C}$, it must be real on $\mathfrak{h}_\R$ and one concludes $\mathfrak{h}_\R\simeq i\mathfrak{t}$.

Therefore, for a Lie algebra $\mathfrak{g}$ which is already compact (not necessarily the compact real form of a complex semisimple Lie algebra) we can consider the Cartan subalgebra given by $i\mathfrak{t}$ and the positive Weyl chamber will be contained in it: $\mathcal{C}^+\subset i\mathfrak{t}$.
\end{remark}

Turning finally again to compact Lie algebras, the result   that will matter the most for our purposes reads:
\begin{proposition} \cite[Lemma 3 and Proposition 3, Chapter 5]{Kir}\\
Every co-adjoint orbit of a compact Lie group $G$ on its dual Lie algebra $\mathfrak{g}^*$ intersects the Cartan subalgebra a number of times equal to the cardinality of the Weyl group of the Root system associated with $\mathfrak{g}^*_\mathbb{C}$. In particular the set of co-adjoint orbits for this action is in one to one correspondence with the positive Weyl chamber $\mathcal{C}^+\subset i\mathfrak{t}^*$ . 

Moreover, given any point $\xi_\bullet\in\mathfrak{t}^*$, its stabiliser is isomorphic to a subgroup containing the maximal torus $T$ of $G$ and there are finitely many orbits, each of which is associated with a subgroup $H$ s.t. $T\subseteq H \subset G$ .
\end{proposition}

This result will allow us to consider all representatives $\xi_\bullet$ of orbits $\orb_{\xi_\bullet}$ in the positive Weyl chamber $\xi_\bullet\in\mathcal{C}^+\subset i\mathfrak{t}^*\simeq i\mathfrak{t}$. Every time an expression on $\orb_{\xi_\bullet}$ is equivariant with respect to the action of $G$, we will be able to compute it on the reference point $\xi_\bullet$. This is what happens for the KKS symplectic form, the SLD and the Fisher Tensor.

\paragraph{Example} Let us work out the main example we will use throughout the main text to fix the ideas. Consider the Lie group $SU(n)$ and its Lie algebra $\mathfrak{su}(n)$ of anti-Hermitian, traceless matrices. Inside $\mathfrak{su}(n)$ we have $\mathfrak{t}$, the toral subalgebra of diagonal, anti-Hermitian, traceless matrices, whose entries are clearly purely imaginary. The complexification of $\mathfrak{su}(n)$ admits a root space decomposition with respect to the Cartan subalgebra $\mathfrak{t}\otimes\mathbb{C}$.

If we pick an element $D_\lambda=\mathrm{diag}\{\lambda_1,\dots,\lambda_n\}\in\mathfrak{t}\otimes\mathbb{C}$, with $\lambda_i\in\mathbb{C}$ and $X_{\alpha_{k,j}}$ is the matrix with a $1$ in the position $(k,j)$ and zero elsewhere we have that the roots for this semisimple Lie algebra are given by $\alpha_{k,j}(D_\lambda)=\lambda_k-\lambda_j$:
\[
[D_\lambda,X_{k,j}]=(\lambda_k-\lambda_j)X_{k,j}.
\]

It is easy to show that the root vectors $H_{\alpha_{k,j}}=[X_{\alpha_{k,j}},X_{-\alpha_{k,j}}]$ are given by diagonal matrices with $+1$ in the $(k,k)$ entry and $-1$ in the $(j,j)$ entry. The real span of $iH_{\alpha_{k,j}}$ is then isomorphic to the subalgebra $\mathfrak{t}$ of traceless, diagonal anti-hermitian matrices, as desired.  

The Weyl group is isomorphic to $S^n$ and permutes the diagonal entries and the roots, as well as the entries of the Torus $U(1)^{n-1}$ maximal in $SU(n)$. A positive Weyl chamber will be given by a choice of ordering and positivity of the roots (which come in pairs of opposite sign). The span of these positive roots is again isomorphic with $i\mathfrak{t}\subset \mathfrak{t}_\mathbb{C}$. 

Notice, finally, that the orbits of $U(n)$ and those of $SU(n)$ are essentially the same. The extra $U(1)$ central subgroup in $U(n)$ gets killed, being trivially a subgroup of any stabiliser (cf. \cite{Kir} and \cite{Hsiang}).

\end{document}